\documentclass[final,1p,times,authoryear]{elsarticle}

\makeatletter
\def\ps@pprintTitle{%
  \let\@oddhead\@empty
  \let\@evenhead\@empty
  \def\@oddfoot{}
  \let\@evenfoot\@oddfoot}
\makeatother

\usepackage{amsmath}
\usepackage{amsthm}
\usepackage{amssymb}
\usepackage{amsfonts}
\usepackage{xifthen}
\usepackage{color}
\usepackage{url}
\usepackage{hyperref}
\hypersetup{
    colorlinks=true,
}
\newtheorem{theorem}{Theorem}[section]
\newtheorem{lemma}[theorem]{Lemma}
\newtheorem{corollary}[theorem]{Corollary}

\newtheorem{proposition}[theorem]{Proposition}
\newtheorem{definition}[theorem]{Definition}

\newtheorem{example}[theorem]{Example}
\newtheorem{remark}[theorem]{Remark}

\newtheorem{fact}[theorem]{Fact}

\numberwithin{equation}{section}

\let\set\mathbb
\let\cal\mathcal
\def\lc{\operatorname{lc}}
\def\shift{\text{S}}
\def\disset{\operatorname{DS}}
\def\bigO{\operatorname{O}}
\def\softO{\operatorname{O^{\sim}}}
\def\lquo{\operatorname{LQ}}
\def\leftquo{\operatorname{LSQ}}
\def\leftrem{\operatorname{LSR}}
\def\hord{\operatorname{hord}}
\def\lord{\operatorname{lord}}

\newcommand{\lmA}[1][]{%
  \ifthenelse{\isempty{#1}}{\cal A_{\lambda,\mu}}{\cal A_{\lambda_{#1},\mu_{#1}}}%
}
\newcommand{\lmphi}[1][]{%
  \ifthenelse{\isempty{#1}}{\phi_{\lambda,\mu}}{\phi_{\lambda_{#1},\mu_{#1}}}%
}
\newcommand{\lmshift}[1][]{%
  \ifthenelse{\isempty{#1}}{\shift_{\lambda,\mu}}{\shift_{\lambda_{#1},\mu_{#1}}}%
}
\begin{document}

\begin{frontmatter}
  
  \title{Efficient Rational Creative Telescoping}
  
\author{Mark Giesbrecht}
\address{Symbolic Computation Group, Cheriton School of Computer Science,
  University of Waterloo,\\
  Waterloo, ON, N2L 3G1, Canada}
\ead{mwg@uwaterloo.ca}

\author{Hui Huang}
\address{School of Mathematical Sciences, Dalian University of Technology,\\
  Dalian, Liaoning, 116024, China}
\ead{huanghui@dlut.edu.cn}

\author{George Labahn}
\address{Symbolic Computation Group, Cheriton School of Computer Science,
  University of Waterloo,\\
  Waterloo, ON, N2L 3G1, Canada}
\ead{glabahn@uwaterloo.ca}

\author{Eugene Zima}
\address{Physics and Computer Science, Wilfrid Laurier University,\\
  Waterloo, ON, N2L 3C5, Canada}
\ead{ezima@wlu.ca}

\begin{abstract}
  We present a new algorithm to compute minimal telescopers for
  rational functions in two discrete variables. As with recent
  reduction-based approaches, our algorithm has the important feature
  that the computation of a telescoper is independent of its
  certificate. In addition, our algorithm uses a compact
  representation of the certificate, which allows it to be easily
  manipulated and analyzed without knowing the precise expanded
  form. This representation hides potential expression swell until the
  final (and optional) expansion, which can be accomplished in time
  polynomial in the size of the expanded certificate. A complexity
  analysis, along with a Maple implementation, indicates that our
  algorithm has better theoretical and practical performance than the
  reduction-based approach in the rational case.
\end{abstract}

\begin{keyword}
  Rational function, GGSZ reduction, Left scalar division with
  remainder, Telescoper
\end{keyword}
\end{frontmatter}

\section{Introduction}\label{SEC:intro}
Creative telescoping is a powerful method pioneered by
\cite{Zeil1990a,Zeil1990b,Zeil1991} in the 1990s and has now become
the cornerstone for finding closed forms for definite sums and
definite integrals in computer algebra. The method mainly constructs a
recurrence (resp.\ differential) equation admitting the prescribed
definite sum (resp.\ integral) as a solution. Employing other
algorithms applicable to the resulting recurrence or differential
equation, it is then possible to find closed form solutions or prove
that there is no such solution. In the latter case, one can still make
use of creative telescoping for such operations as determining
asymptotic expansions of the sum or integral under investigation.

In the case of summation, in order to compute a sum of the form
$\sum_{y=a}^bf(x,y)$, the main task of creative telescoping consists
of constructing polynomials $c_0,\dots, c_\rho$ in $x$, not all zero,
and another function $g$ in the same domain as $f$ such that
\begin{equation}\label{EQ:ct}
  c_\rho(x) f(x+\rho,y) +\cdots+c_1(x)f(x+1,y)+c_0(x)f(x,y)=
  g(x,y+1)-g(x,y).
\end{equation}
The number $\rho$ may or may not be part of the input. If $c_0, \dots,
c_\rho$ and $g$ are as above, then we say that
$L=c_\rho\shift_x^\rho+\cdots+c_1\shift_x+c_0$ with $\shift_x$ being
the shift operator in $x$ is a {\em telescoper} for $f$ and $g$ is a
{\em certificate} for $L$. If $c_\rho \neq 0$ then the integer $\rho$
is the {\em order} of $L$. Finally, the maximum degree in $x$ among
the polynomials $c_\ell$ is the {\em degree} of $L$.

The technique of creative telescoping has seen various generalizations
and improvements over the past three decades.  At the present time,
the so-called reduction-based approach originating
from~\citep{BCCL2010} has drawn the most attention, as it is both
efficient in practice and equipped with the useful feature that it
allows one to find a telescoper without necessarily also computing the
corresponding certificate. In other words, the computation of the
$c_\ell$ in \eqref{EQ:ct} is separated from the computation of $g$. In
a typical situation where the size of the $c_\ell$ is much smaller
than the size of $g$ and the right-hand side of \eqref{EQ:ct}
collapses to zero when summing over the defining interval, this
approach enables one to merely compute the $c_\ell$ avoiding the
costly yet unnecessary computation of the certificate $g$.  In
applications where a certificate is required, the approach also allows
one to express the certificate as an unnormalized sum so that the
summands are concatenated symbolically without actually calculating
the sum. These summands are often of much smaller sizes than the
original certificate. So far, the reduction-based approach has been
worked out for many special functions. We refer to \citep{Chen2019}
for an excellent exposition of all these algorithms.

However, it is also the case that the unnormalized expression for the
certificate returned by the reduction-based approach can introduce
superfluous terms which eventually cancel out when normalized. These
terms will not contribute to the final output but will increase sizes
of intermediate results and thus deteriorate the performance of the
approach in these applications. In order to illustrate this issue, let
us consider a simple discrete rational function of the form
\begin{equation}\label{EQ:ex}
  f(x,y) = \frac{x}{x+3y + 3m}-\frac{x}{x+3 y+3}
  +\frac{x}{x+3 y},
\end{equation}
where $m$ is an integer greater than one. Applying a reduction method,
for example, in \citep{Abra1975}, to the given rational function $f$
yields
\begin{equation}\label{EQ:exred0}
  f(x,y) = g_0(x,y+1)-g_0(x,y) + r_0\quad
  \text{with}\ g_0(x,y)=\sum_{k=1}^{m-1}\frac{x}{x+3y+3k} ~~
  \ \text{and}\ ~~ r_0=\frac{x}{x+3 y},
\end{equation}
where $r_0$ has the denominator of lowest possible degree in $y$.
Based on the form \eqref{EQ:exred0}, iteratively applying the chosen
reduction method to each $f(x+\ell,y)$ for $\ell \geq 0$ gives
\begin{align*}
  f(x+\ell,y) = g_\ell(x,y+1)-g_\ell(x,y)+r_\ell
  \quad\text{with}\ r_\ell=\frac{x+\ell}{x+3 y+\bar\ell},
\end{align*}
where $\bar\ell\in\{0,1,2\}$ is $\ell$ reduced modulo $3$ and
\[
g_\ell(x,y) = g_0(x+\ell,y)+\sum_{k=1}^{\lfloor \ell/3\rfloor}
\frac{x+\ell}{x+3 y +3(k-1) + \bar \ell}.
\]
Finding a linear dependency amongst the $r_\ell$ reduces to solving
the following linear system
\begin{equation}\label{EQ:exrctsys}
\begin{pmatrix}
  9x & 9x+9 & 9x+18 & 9x+27\\
  6x^2+9x & 6x^2+12x+6 & 6x^2+15x+6 & 6x^2+27x+27\\
  x^3+3x^2+2x & x^3+3x^2+2x & x^3+3x^2+2x & x^3+6x^2+11x+6
\end{pmatrix}
\begin{pmatrix} c_0\\c_1\\c_2\\c_3\end{pmatrix}
  =\begin{pmatrix} 0\\0\\0\end{pmatrix}.
\end{equation}
A nontrivial polynomial solution $(c_0,c_1,c_2,c_3) =(-(x+3),0,0,x)$
then gives
\begin{equation}\label{EQ:extel}
  L = x\,\shift_x^3 - (x+3),
\end{equation}
a telescoper for $f$ of minimal order with a corresponding certificate
\begin{equation}\label{EQ:excert}
  g(x,y) = x\cdot g_3(x,y) - (x+3)\cdot g_0(x,y)
  =\frac{x(x+3)}{x+3 y +3m}
  -\frac{x(x+3)}{x+3 y +3}
  +\frac{x(x+3)}{x+3 y}
\end{equation}
obtained by canceling out the common $m-2$ terms in the summation. As
the $m$ increases, the size of each $g_\ell$ grows rapidly, whereas
the expanded certificate $g$ may still be small. In this particular
example, it is actually more reasonable to use the decomposition
\[
f(x,y) = g_0(x,y+1)-g_0(x,y) + r_0,\quad
\text{with}\ g_0(x,y)=-\frac{x}{x+3 y}
\ \text{and} \ r_0 = \frac{x}{x+3 y+3m},
\]
instead of \eqref{EQ:exred0}. This leads to an alternate choice of
$r_\ell$ for each $f(x+\ell,y)$, with the corresponding $g_\ell$
having the denominator of much smaller degree in $y$.  With this
choice one gets the same telescoper $L$ and the same certificate $g$
as before, but this time there is no cancellation happening in
\eqref{EQ:excert}. That is, the unnormalized sum gives the final size
of the certificate.  This suggests a solution to the above issue.
Namely, find an initial decomposition \eqref{EQ:exred0} with both
$r_0$ and $g_0$ having denominators of lowest possible degrees in $y$
using the method proposed in \citep{Poly2011,Zima2011} to initiate the
iterative process of the reduction-based approach. However this
process requires a full irreducible factorization of a polynomial.

Separate from the previously mentioned work, there is an alternate
method developed by \cite{Le2003a} which constructs telescopers in a
direct fashion.  This method was later used by \cite{ChKa2012a} to
obtain the best order-degree curve known so far for telescopers of
bivariate rational functions. Currently, the method has only been
worked out for bivariate rational functions in the ($q$-)shift case.
Nevertheless, the method is still interesting because it also has the
feature that the computation of a telescoper does not depend on its
certificate. In order to demonstrate its main idea, consider again the
rational function $f$ given in \eqref{EQ:ex}. As with the
reduction-based approach, this method first decomposes $f$ as in
\eqref{EQ:exred0}.  The difference is that it later decomposes $r_0$
as the sum of several simple fractions of numerators in $x$ only,
which in our example is merely $x\left(\frac1{x+3y}\right)$. By
viewing $x=x\,\shift_x^0$ as a recurrence operator of order zero and
using the fact that $\shift_x^3-1$ is a minimal telescoper for
$\frac1{x+3y}$ with a corresponding certificate $\frac1{x+3y}$, Le's
method then computes the least common left multiple of $x$ and
$\shift_x^3-1$ with the left cofactor of $x$ (resp.\ $\shift_x^3-1$)
giving rise to the same telescoper $L$ as in \eqref{EQ:extel}
(resp.\ its certificate $\frac{x(x+3)}{x+3y}$) for the simple fraction
$x\left(\frac1{x+3y}\right)=r_0$. In the more general case where there
is more than one simple fraction in $r_0$, one finds a telescoper of
minimal order for $r_0$ by calculating the least common left multiple
of all telescopers for individual simple fractions. Together with
\eqref{EQ:exred0}, the method yields a telescoper of minimal order for
$f$, namely $L$, as well as its (optional) certificate of the form
\[
g = L(g_0) + \frac{x(x+3)}{x+3y}.
\]
Rather than leaving the certificate as a (potentially large)
unnormalized sum as done by the reduction-based approach, this method
represents the certificate by recurrence operators.  This
representation enables one to more easily manipulate the certificate
or analyze its various properties such as the singularities without
knowing its expanded form. However, the intermediate expression swell
which happens in the certificate is still unavoidable due to
\eqref{EQ:exred0}. A second disadvantage is that this method requires
the numerator of each simple fraction appearing in the decomposition
to be independent of $y$, often requiring one to work in algebraic
extensions of the base field.

\subsection{Proposed new approach}
Our new algorithm constructs a telescoper for a rational function in a
similar fashion as the reduction-based approach, but incorporating the
idea from the method of \cite{Le2003a}. As a result, our algorithm
completely avoids algebraic extensions of the base field and
intermediate expression swell in the certificate.  In order to
describe the main idea of our algorithm, let us continue the example
\eqref{EQ:ex}. Unlike the reduction-based approach and the method of
Le, we first find a recurrence operator $M$ allowing us to rewrite $f$
in the form
\[
f = \underbrace{(x\,\shift_x^{3m}-x\,\shift_x^3 +x)}_{M}\left(\frac1{x+3y}\right).
\]
Assume that we want to find a telescoper for $f$ of order no more than
$\rho\in\set N$, say $\rho = 3$. We then make an ansatz $L =
c_3\shift_x^3+c_2\shift_x^2+c_1\shift_x+c_0$ with $c_0,c_1,c_2,c_3$ to
be determined.  Using the relation $\shift_x^3(x+3y) = \shift_y(x+3y)$
with $\shift_y$ being the shift operator in $y$, we calculate the left
scalar remainder
\[
R = (x+2)c_2 \shift_x^2+(x+1)c_1\shift_x+((x+3)c_3 + xc_0)
\]
from the so-called left scalar division of $L\odot M$ by $\shift_y-1$,
where $L\odot M$ is the multiplication of $M$ by $L$ from the
left-hand side modulo the left ideal generated by
$\shift_y-\shift_x^3$ (see Section~\ref{SEC:integerlinear} for a
precise definition). We show that $L$ is a telescoper if and only if
$R = 0$. The problem is then reduced to solving the following linear
system
\begin{equation}\label{EQ:exoctsys}
\begin{pmatrix}
  x & 0 & 0 & x+3\\
  0 & 0 & x+2 & 0\\
  0 & x+1 & 0 & 0
\end{pmatrix}
\begin{pmatrix}
  c_0\\c_1\\c_2\\c_3\end{pmatrix}
  =\begin{pmatrix} 0\\0\\0\end{pmatrix}.
\end{equation}
One immediately reads a nontrivial polynomial solution
$(c_0,c_1,c_2,c_3) =(-(x+3),0,0,x)$, which yields the telescoper $L$
given by \eqref{EQ:extel}. In terms of the certificate, we either follow
the idea from \citep{GGSZ2003} and use the compact representation
\[
g = \leftquo(L\odot M, \shift_y-1)\left(\frac{1}{x+3y}\right),
\]
or expand it as \eqref{EQ:excert} by noticing $\leftquo(L\odot M,
\shift_y-1) =x(x+3)\shift_x^{3m}-x(x+3)\shift_x^3 + x(x+3)$, where
$\leftquo$ denotes the left scalar quotient obtained from the left
scalar division.

In the case where the induced linear system admits no nontrivial
solutions, we then have shown that there does not exist any telescoper
of order no more than $\rho$ for the given rational function.  In
order to find a telescoper of minimal order, one can execute the above
process incrementally by letting $\rho = 0,1, 2,\dots$.  The
termination of the new algorithm is guaranteed by the existence
criterion for telescopers of rational functions given in
\citep[Theorem~1]{AbLe2002}, which essentially boils down to checking
the integer-linearity of a polynomial. In the general case, the
operator $\shift_x$ in $M$ is replaced by a special recurrence
operator acting particularly on integer-linear rational functions of
one type, and the given rational function is initially separated into
several simple fractions according to integer-linear types.

In summary, our main contribution is a new algorithm for computing
minimal telescopers for rational functions. As with the
reduction-based approach and the method of Le, our algorithm separates
the computation of the telescoper from that of the certificate. When
the certificate is needed our algorithm computes it in a compact form,
hiding potential expression swell until a final, optional expansion.
Compared to Le's method, our algorithm avoids the need for algebraic
extensions.  In addition, if an expanded form for the certificate is
desired then it can be computed easily in time polynomial in the size
of the expanded certificate. Moreover, comparing \eqref{EQ:exoctsys}
with \eqref{EQ:exrctsys} suggests that our algorithm also has better
control for the size of intermediate expressions involved in the
computation of the telescoper.

The arithmetic cost of our new algorithm, as well as that of the
reduction-based approach in the rational case, is analyzed in this
paper. We note that, until recently, most complexity analyses were
done for the differential
case~\citep{BCCL2010,BLS2013,BCLS2018,vdHo2020} whereas little has
been known for the shift case. The complexity analysis shows that our
new algorithm is at least one order of magnitude faster than the
reduction-based approach in the rational case when the certificate is
not expanded. A Maple implementation further confirms that our
approach outperforms the reduction-based approach when restricted to
the rational case. In addition, the new algorithm is easy to analyze
and leads to a tight order-degree curve for telescopers, a property
shared with the method of Le.

The remainder of the paper proceeds as follows. Some basic notions and
results are recalled in the next section for later use. In particular,
two important decompositions of polynomials in the bivariate setting
are reviewed. A kind of recurrence operators specifically working on
integer-linear rational functions of one type is introduced in
Section~\ref{SEC:integerlinear}.  Based on basic arithmetic for
operators of this kind, Section~\ref{SEC:oct} describes a new
algorithm to construct a telescoper of minimal order for bivariate
rational functions. Section~\ref{SEC:complexity} provides a cost
analysis of our new algorithm, followed in Section~\ref{SEC:rct} by a
brief summary and a cost analysis of the reduction-based approach in
the rational case.  Section~\ref{SEC:timing} contains some
experimental comparison among all above-mentioned approaches. The
paper ends with some topics for future research.

\section{Preliminaries}\label{SEC:prel}
Throughout the paper $\set K$ denotes a field of characteristic zero
with $\set K(x,y)$ the field of rational functions in~$x,y$ over~$\set
K$. We let $\sigma_x$ and $\sigma_y$ be the automorphisms over $\set
K(x,y)$, which, for any $f\in \set K(x,y)$, are defined by
\[
\sigma_x(f(x,y))=f(x+1,y)\quad\text{and}\quad
\sigma_y(f(x,y))=f(x,y+1).
\]
A rational function $f\in \set K(x,y)$ is called {\em summable} with
respect to~$y$ (or {\em $\sigma_y$-summable} for short) if $f=
\sigma_y(g)-g$ for some $g\in \set K(x,y)$. A nonzero polynomial $f\in
\set K[x,y]$ is called {\em shift-free} with respect to~$y$ (or {\em
  $\sigma_y$-free} for short) if $\gcd(f,\sigma_y^\ell(f)) \in \set
K[x]$ for all nonzero integers $\ell$.

Let $f$ be a polynomial in $\set K[x,y]$. Throughout this paper, we
will order terms using a pure lexicographic order with $x\prec y$. For
this order, we let $\lc_{x,y}(f)$ denote the leading coefficient of
$f$ over $\set K$ with respect to $x,y$. We say that $f$ is {\em
  monic} with respect to $x,y$ if $\lc_{x,y}(f)=1$. In the sequel,
unless there is a danger of confusion, we will just say that $f$ is
monic, omitting the variables. We also denote by $\deg_x(f)$ and
$\deg_y(f)$ the degrees of $f$ with respect to $x$ and $y$,
respectively, following the convention that
$\deg_x(0)=\deg_y(0)=-\infty$.

Let $\set K(x,y)[\shift_x,\shift_y]$ be the ring of linear recurrence
operators in $x,y$ over $\set K(x,y)$, in which the following
commutation rules hold: $\shift_x\shift_y=\shift_y\shift_x$ and
$\shift_x f=\sigma_x(f)\shift_x$, $\shift_yf=\sigma_y(f)\shift_y$ for
any $f\in \set K(x,y)$. The application of an operator $L =
\sum_{i,j\geq 0} a_{ij}\shift_x^i\shift_y^j$ in $\set
K(x,y)[\shift_x,\shift_y]$ to a rational function $f\in \set K(x,y)$
is then defined as $L(f)=\sum_{i,j\geq 0}a_{ij}\sigma_x^i\sigma_y^j(f)$.
\begin{definition}\label{DEF:telescoper}
  Let $f$ be a rational function in $\set K(x,y)$. A nonzero operator
  $L\in \set K[x][\emph\shift_x]$ is called a {\em telescoper} for $f$
  if $L(f)$ is $\sigma_y$-summable, or equivalently, there exists a
  rational function $g\in \set K(x,y)$ such that
  \[
  L(f) = (\emph\shift_y-1)(g),
  \]
  where 1 denotes the identity map of $\set K(x,y)$. We call $g$ a
  corresponding {\em certificate} for $L$. The {\em order} and
  {\em degree} of $L$ are defined to be its degree in $\emph\shift_x$
  and the maximum degree in $x$ of its coefficients with respect to
  $\emph\shift_x$, respectively. A telescoper of minimal order is also
  called a {\em minimal telescoper}.
\end{definition}

In the rest of this section, we introduce two important decompositions
of polynomials, both of which will play crucial roles in our later
algorithms.

\subsection{Shift-homogeneous decomposition and GGSZ reduction}
\label{SUBSEC:ggszred}
Recall that two polynomials $f,g\in \set K[x,y]$ are called 
{\em shift-equivalent} with respect to~$y$ (or {\em $\sigma_y$-equivalent} 
for short), denoted by $f\sim_y g$, if $f = \sigma_y^m(g)$ for some 
$m\in \set Z$. Clearly, $\sim_y$ is an equivalence relation. The 
$\sigma_y$-equivalence of two polynomials can be easily recognized
by comparing coefficients.

By grouping together its $\sigma_y$-equivalent irreducible factors,
any polynomial $g\in \set K[x,y]$ can be written in the form
\begin{equation}\label{EQ:shiftlessdecomp}
g = c\,\prod_{i=1}^m\prod_{j=1}^{n_i}\sigma_y^{\nu_{ij}}(g_i)^{e_{ij}},
\end{equation}
where $c\in \set K[x]$, $m,n_i,\nu_{ij}, e_{ij}\in\set N$ with
$0=\nu_{i1}<\nu_{i2}<\dots<\nu_{in_i}$ and $e_{ij}>0$, $g_i \in \set
K[x,y]$ is monic, irreducible and of positive degree in $y$, and the
$g_i$ are pairwise $\sigma_y$-inequivalent. Since $\set K[x,y]$ is a
unique factorization domain, the decomposition
\eqref{EQ:shiftlessdecomp} is unique up to the order of factors.  In
view of this, we call \eqref{EQ:shiftlessdecomp} the {\em
  shift-homogeneous decomposition} of $g$ with respect to~$y$.

We note that in the context of univariate polynomials, the
shift-homogeneous decomposition is the same as the most refined
shiftless decomposition defined in \citep{GGSZ2003}. 
Based on shiftless decompositions, a reduction algorithm for
univariate rational functions, named {\bf RatSum}, was developed 
in the same paper. This
algorithm can be carried over to the case of bivariate rational
functions in a straightforward manner, to which we will refer as the
GGSZ reduction later for convenience, named after the authors.  The
input and output of the GGSZ reduction are given below.

\smallskip\noindent{\bf GGSZReduction.}  Given a rational function
$f\in \set K(x,y)$, compute two rational functions $h, r$ in $\set
K(x,y)$ with $r=a/b$, $a,b\in \set K[x,y]$, $\deg_y(a)<\deg_y(b)$ and
$b$ being $\sigma_y$-free such that
\begin{equation}\label{EQ:ggszred}
f = (\shift_y-1)(h)+r.
\end{equation}

Such a reduction algorithm is vital for many creative telescoping
approaches, including the reduction-based one in \citep{CHKL2015}, the
method of \cite{Le2003a} and the algorithm introduced in this paper.
Unlike previous reduction algorithms as given in
\citep{Abra1975,Paul1995}, the GGSZ reduction uses a compact
representation of $h$ in \eqref{EQ:ggszred} in terms of left quotients
(see Example~\ref{EX:ggszred} for an illustration), and hence works in
polynomial-time of the size of the input without the final expansion.

\begin{example}\label{EX:ggszred}
  Let $g$ be a polynomial of the form
  \[
  (xy+1)(x(y+1)+1)(x(y+29)+1)(x(y+30)+1)
  ((-5x+2y)^2+1)((-5x+2y+1)^2+1)((3x+10y)^3+1).
  \]
  Then by grouping together $\sigma_y$-equivalent irreducible factors,
  we obtain
  \[
  g = g_0 \sigma_y(g_0)\sigma_y^{29}(g_0)\sigma_y^{30}(g_0)
  g_1 g_2 g_3 g_4,
  \]
  where $g_0 =xy+1$, $g_1 = (-5x+2y)^2+1$, $g_2 = (-5x+2y+1)^2+1$
  $g_3 = (3x+10y)+1$ and $g_4 = (3x+10y)^2-(3x+10y)+1$.  Up to making
  $g_1,g_2,g_3,g_4$ monic, the above equation gives the
  shift-homogeneous decomposition of $g$ with respect to~$y$.
	
  Let $f$ be a rational function with denominator $g$ admitting the
  following decomposition
  \begin{align*}
    \frac{2x+3}{\sigma_y^{30}(g_0)}-\frac{2x+3}{\sigma_y^{29}(g_0)}
    -\frac1{\sigma_y(g_0)}+\frac1{g_0}+
    \frac{2x^2+1}{(-5x+2y)^2+1}+\frac{x-1}{(-5x+2y+1)^2+1}
    +\frac{xy+1}{(3x+10y)^3+1}.
  \end{align*}
  We remark that all decomposed forms given in our examples are for
  readability only. Applying the GGSZ reduction to $f$ then yields
  \eqref{EQ:ggszred} with
  \begin{align}
    h &= \lquo((2x+3)\emph\shift_y^{30}-(2x+3)\emph\shift_y^{29}-\emph\shift_y+1,
    \emph\shift_y-1)\left(\frac{1}{g_0}\right)
    =((2x+3)\emph\shift_y^{29}-1)\left(\frac{1}{g_0}\right)\nonumber\\
    \quad\text{and}\quad
    r &= \frac{2x^2+1}{(-5x+2y)^2+1}+\frac{x-1}{(-5x+2y+1)^2+1}
    +\frac{xy+1}{(3x+10y)^3+1},\label{EQ:red0}
  \end{align}
  where $\lquo$ denotes the left quotient in the ring $\set
  Q(x,y)[\emph\shift_y]$. Note that, in this example, the left
  quotient in $h$ is a sparse operator although it is of relatively
  high order 29. Hence the expanded form of $h$ is small. Since $r\neq
  0$, then $f$ is not $\sigma_y$-summable by
  \citep[Theorem~12]{GGSZ2003}. We will use $f$ as a running example
  in this paper.
\end{example}

\subsection{Integer-linear decomposition and its refinement}
Recall that an irreducible polynomial $g\in \set K[x,y]$ is called
{\em integer-linear} (over~$\set K$) if it is of the form $p(\lambda
x+\mu y)$ for some integers $\lambda,\mu$ and a univariate polynomial
$p\in \set K[z]$.  Note that $\lambda,\mu$ cannot both be zero since
$g$ is irreducible and thus nonunit. By pulling out a common factor
and absorbing it into $p$, one may assume without loss of generality
that $\lambda,\mu$ are coprime and that $\mu\geq 0$.  Such a pair
$(\lambda,\mu)$ is unique and is called the {\em integer-linear type}
of $g$. For the sake of completeness, we let a constant polynomial be
integer-linear of type $(0,0)$.  A polynomial in $\set K[x,y]$ is then
called {\em integer-linear} (over~$\set K$) if all its irreducible
factors are integer-linear, possibly with different integer-linear
types.  A rational function in $\set K(x,y)$ is called {\em
  integer-linear} (over~$\set K$) if its denominator and numerator are
both integer-linear.
\begin{definition}\label{DEF:ildecomp}
  Let $g\in \set K[x,y]$ be a polynomial admitting the decomposition
  \begin{equation}\label{EQ:ildecomp}
    g = p_0(x,y) \prod_{i=1}^{m} p_i(\lambda_i x+ \mu_i y),
  \end{equation}
  where $p_0\in \set K[x,y]$, $m\in \set N$, $\lambda_i,\mu_i\in \set
  Z$ and $p_{i}\in \set K[z]$ for $1\leq i\leq m$. Then
  \eqref{EQ:ildecomp} is called the {\em integer-linear decomposition}
  of $g$ if
  \begin{itemize}
  \item none of irreducible factors of $p_0$ is integer-linear;
  \item $p_1,\dots,p_m$ are monic and of positive degrees in $z$;
  \item each $(\lambda_i,\mu_i)$ satisfies $\gcd(\lambda_i,\mu_i) = 1$
    and $\mu_i\geq 0$;
  \item any two pairs of the $(\lambda_i,\mu_i)$ are distinct.
  \end{itemize}
  The $(\lambda_i,\mu_i)$ are called {\em integer-linear types} of
  $g$. If $g$ is clear from the context, we will simply say that the
  $(\lambda_i,\mu_i)$ are integer-linear types.
\end{definition}
Clearly, $g$ is integer-linear if and only if $p_0\in \set K$ in
\eqref{EQ:ildecomp}.  By the uniqueness of full factorization and
integer-linear types, we see that every polynomial admits a unique
integer-linear decomposition up to the order of the factors.

In terms of computation, an efficient algorithm for finding
integer-linear decompositions of general multivariate polynomials was
recently proposed by authors~\citep{GHLZ2019}. Compared with previous
known approaches \citep{AbLe2002,LiZh2013}, this algorithm performs
better both in theory and in practice.

Recall that two polynomials $f,g\in \set K[x,y]$ are called
{\em shift-equivalent} with respect to~$x,y$ (or
{\em $(\sigma_x,\sigma_y)$-equivalent} for short), denoted by $f
\sim_{x,y} g$, if there exist $\ell,m\in \set Z$ such that
$f=\sigma_x^{\ell}\sigma_y^{m}(g)$. Clearly, $\sim_{x,y}$ is an
equivalence relation and contains the relation $\sim_y$.
Suppose that $f,g$ are integer-linear of the forms
$f(x,y)=p_1(\lambda_1 x+\mu_1 y)$ and $g(x,y)=p_2(\lambda_2 x+\mu_2y)$
for $p_i\in \set K[z]$ and $\lambda_i,\mu_i\in \set Z$ with $\mu_i\geq
0$ and $\gcd(\lambda_i,\mu_i)=1$. Then $f\sim_{x,y} g$ implies that
$(\lambda_1,\mu_1)=(\lambda_2,\mu_2)$ and $p_1(z)=p_2(z+\ell)$ for
some $\ell\in \set Z$, and conversely.
This indicates that for any two integer-linear polynomials of single
types, testing their $(\sigma_x,\sigma_y)$-equivalence amounts to
checking the equality of the integer-linear types and the
shift-equivalence of univariate polynomials.

Let $p\in\set K[z]$ be a monic polynomial of positive degree in $z$, 
and let $(\lambda,\mu)$ be an integer-linear type with $\mu > 0$.  
By computing the shift-homogeneous decomposition of $p$
with respect to $z$, we obtain
$p(z) = \prod_{i=1}^m\prod_{j=1}^{n_i}p_i(z+\nu_{ij})^{e_{ij}}$,
where $m,n_i,\nu_{ij}, e_{ij}\in\set N$ with
$0=\nu_{i1}<\nu_{i2}<\dots<\nu_{in_i}$ and $e_{ij}>0$, $p_i\in\set K[z]$
is monic and irreducible, and the $p_i$ are pairwise
shift-inequivalent with respect to $z$. It then follows that
\[
p(\lambda x + \mu y) = \prod_{i=1}^m\prod_{j=1}^{n_i}
p_i(\lambda x + \mu y + \nu_{ij})^{e_{ij}},
\]
where the $p_i(\lambda x+ \mu y)$ are pairwise
$(\sigma_x,\sigma_y)$-inequivalent.

Consider now a polynomial $g\in\set K[x,y]$ with the integer-linear
decomposition \eqref{EQ:ildecomp}. For each factor $p_i(\lambda_i
x+\mu_i y)$ with $1\leq i\leq m$ in \eqref{EQ:ildecomp}, if $\mu_i =
0$ we then absorb it into $p_0$; otherwise we further split it into
distinct $(\sigma_x,\sigma_y)$-equivalence classes using the procedure
described in the preceding paragraph. By relabeling all the resulting
factors, we finally derive the following decomposition (with a slight
abuse of notation)
\begin{equation}\label{EQ:refinedildecomp}
  g=p_0(x,y)\prod_{i=1}^m\prod_{j=1}^{n_i} p_{i}(\lambda_i x+\mu_i
  y+\nu_{ij})^{e_{ij}},
\end{equation}
where $p_0\in \set K[x,y]$, $m,n_i,\nu_{ij}, e_{ij}\in \set N$,
$\lambda_i, \mu_i\in\set Z$ and $p_1,\dots,p_m\in \set K[z]$
satisfying
\begin{itemize}
\item none of irreducible factors of $p_0$ of positive degree in $y$ 
  is integer-linear;
\item $p_1,\dots,p_m$ are monic and irreducible;
\item each $(\lambda_i,\mu_i)$ is an integer-linear type with $\mu_i>0$;
\item $p_i(\lambda_i x + \mu_i y)\nsim_{x,y}p_j(\lambda_j x + \mu_j y)$
  for any two integers $i,j$ with $1\leq i< j\leq m$; or equivalently,
  either $(\lambda_i,\mu_i)\neq(\lambda_j,\mu_j)$ or $p_i$ is
  shift-inequivalent with $p_j$ with respect to $z$;
\item $0=\nu_{i1}<\cdots <\nu_{in_i}$ and $e_{ij}>0$.
\end{itemize}
Evidently, the above decomposition is unique up to the order of
factors.  We will call \eqref{EQ:refinedildecomp} the
{\em refined integer-linear decomposition} of the polynomial $g$.

\begin{example}\label{EX:intlinear}
  Let $g$ be the same polynomial as given in Example~\ref{EX:ggszred}.
  By definition, it is easy to see that $g$ possesses the
  integer-linear decomposition
  \[
  g = p_0(x,y) \tilde p_1(-5x+2y) \tilde p_2(3x+10y),
  \]
  where $p_0=g_0\sigma_y(g_0)\sigma_y^{29}(g_0)\sigma_y^{30}(g_0)$
  with $g_0 = x y + 1$, $\tilde p_1(z) = (z^2+1)((z+1)^2+1)$ and
  $\tilde p_2(z) = z^3+1$. Computing the shift-homogeneous
  decompositions of $\tilde p_1$ and $\tilde p_2$ with respect to $z$
  then yields the refined integer-linear decomposition
  \begin{equation}\label{EQ:grefined}
    g = p_0(x,y)p_1(-5x+2y)p_1(-5x+2y+1)p_2(3x+10y)p_3(3x+10y)
  \end{equation}
  with $p_1(z) = z^2+1$, $p_2(z) = z+1$ and $p_3(z) = z^2-z+1$.
\end{example}

\section{Integer-linear operators}\label{SEC:integerlinear}
In this section, we introduce another vital ingredient of our
algorithms, in this case a special recurrence operator specifically
acting on integer-linear rational functions of a single type.

By a standard localization at a left Ore set (see \citep[\S 0.9]{Cohn1985}
or \citep[\S 3.1]{Rowe1988a}), the ring $\set K(x,y)[\shift_x,\shift_y]$
can be extended to
\[
\cal A:=\set K(x,y)[\shift_x, \shift_y, \shift_x^{-1}, \shift_y^{-1}].
\]
Here $\shift_x^{-1} f = \sigma_x^{-1}(f) \shift_x^{-1}$ and
$\shift_y^{-1} f = \sigma_y^{-1}(f) \shift_y^{-1}$ for all $f\in\set
K(x,y)$ with $\sigma_x^{-1}, \sigma_y^{-1}$ denoting the inverse maps
of the automorphisms $\sigma_x,\sigma_y$, respectively. For an
operator $L\in\cal A$, there exist unique rational functions
$a_{ij}\in\set K(x,y)$, finitely many nonzero, such that $L =
\sum_{i,j \in \set Z} a_{ij} \shift_x^i \shift_y^j$.

In the rest of this section, we fix a pair $(\lambda,\mu)$ of coprime
integers with $\mu>0$. Then there exist unique integers $\alpha,\beta$
such that
\begin{equation}\label{EQ:bezout}
  \alpha \lambda+\beta\mu=1,
\end{equation}
with the constraints $0\leq \alpha <\mu$ and $|\beta| \leq |\lambda|$
if $\lambda \neq 0$, or $\alpha = 0$ and $\beta = 1$ otherwise. Set
$\lmshift$ to be the product $\shift_x^\alpha\shift_y^\beta$. Then
\[
\lmA := \set K(x,y)[\lmshift,\lmshift^{-1}]
\]
is a subring of $\cal A$, which consists of all integer-linear
operators of type $(\lambda,\mu)$.

We can view $\lmA$ as a left module over $\cal A$ as follows. Define
the left $\set K(x,y)$-linear map
\[ 
\begin{array}{cccc}
  \lmphi: & \cal A & \longrightarrow & \lmA \\
  &  \sum_{i,j \in \set Z} a_{ij} \shift_x^i \shift_y^j & \mapsto 
  & \sum_{i,j \in \set Z} a_{ij} \lmshift^{i \lambda + j \mu}. 
\end{array}
\]
The image and kernel of $\lmphi$ are determined below.
\begin{proposition}\label{PROP:surjective}
  The restriction of $\lmphi$ on $\lmA$ is the identity. Consequently,
  $\lmphi$ is surjective.
\end{proposition}
\begin{proof}
  By \eqref{EQ:bezout}, $\lmphi(\lmshift^i) = \lmshift^i$ for all
  $i\in \set Z$, which, together with the definition of $\lmphi$,
  implies the assertion.
\end{proof}
\begin{lemma}\label{LEM:prod}
  For every $L\in \cal A$ and $k,\ell\in\set Z$, we have
  $\lmphi(L\emph\shift_x^k\emph\shift_y^\ell) =
  \lmphi(L)\lmphi(\emph\shift_x^k\emph\shift_y^\ell)$.
\end{lemma}
\begin{proof}
  A straightforward calculation based on the definition of $\lmphi$
  implies that
  \[
  \lmphi((\shift_x^i\shift_y^j)(\shift_x^k\shift_y^\ell))
  = \lmphi(\shift_x^i\shift_y^j)\lmphi(\shift_x^k\shift_y^\ell)
  \quad\text{for all $i,j\in\set Z$}.
  \]
  The lemma then follows from the linearity of $\lmphi$.
\end{proof}
The above lemma does not imply that $\lmphi$ is a ring homomorphism.
In fact, one can easily verify that $\lmphi(\shift_y y)\neq
\lmphi(\shift_y)\lmphi(y)$ provided that $\beta\mu\neq 1$.
\begin{proposition}\label{PROP:kernel}
  The kernel of $\lmphi$ is the left ideal generated by
  $\emph\shift_x-\emph\shift_{\lambda,\mu}^\lambda$ and
  $\emph\shift_y-\emph\shift_{\lambda,\mu}^\mu$ in $\cal A$.
\end{proposition}
\begin{proof}
  Let $I$ be the left ideal generated by $\shift_x-\lmshift^\lambda$
  and $\shift_y-\lmshift^\mu$ in $\cal A$. For any $L\in I$, there are
  $P,Q\in\cal A$ such that $L =
  P(\shift_x-\lmshift^\lambda)+Q(\shift_y-\lmshift^\mu)$.  By
  Lemma~\ref{LEM:prod}, $\lmphi(L) =
  \lmphi(P)\lmphi(\shift_x-\lmshift^\lambda) +
  \lmphi(Q)\lmphi(\shift_y-\lmshift^\mu)$. It follows from the
  definition of $\lmphi$ and Proposition~\ref{PROP:surjective} that
  $\lmphi(\shift_x-\lmshift^\lambda) = \lmphi(\shift_y-\lmshift^\mu) =
  0$, and so also $\lmphi(L) = 0$. We have that $I \subseteq
  \ker(\lmphi)$.
	
  Conversely, we first observe that every $L\in\cal A$ can be
  decomposed as $L = M + R$ for some $M\in I$ and $R\in\lmA$. This is
  because every monomial $\shift_x^i\shift_y^j$ in $L$ with
  $i,j\in\set Z$ can be rewritten as
  $(\shift_x-\lmshift^\lambda+\lmshift^\lambda)^i
  (\shift_y-\lmshift^\mu+\lmshift^\mu)^j$ and
  $(\shift_x-\lmshift^\lambda),\lmshift^\lambda,(\shift_y-\lmshift^\mu),
  \lmshift^\mu$ multiplicatively commute with each other, so expanding
  the powers yields the desired result.  Then $\lmphi(L) = \lmphi(R)$
  since $M \in I\subseteq \ker(\lmphi)$.  Moreover, $\lmphi(L) = R$ by
  Proposition~\ref{PROP:surjective}.  We see that $L\in\ker(\lmphi)$
  implies $R = 0$. Hence $\ker(\lmphi)\subseteq I$.
\end{proof}
According to Proposition~\ref{PROP:surjective}, $\cal A/\ker(\lmphi)$
is isomorphic to $\lmA$ as additive groups. Furthermore, $\cal
A/\ker(\lmphi)$ is a left module over $\cal A$ by
Proposition~\ref{PROP:kernel}.  Hence, $\lmA$ can be viewed as a left
module over $\cal A$ as well. Its left scalar multiplication is
defined via $\lmphi$ as follows. For all $L\in \cal A$ and $M \in
\lmA$, the result obtained by multiplying $L$ from the left-hand side
to $M$ is $\lmphi(LM)$, which is denoted by $L\odot M$ when the pair
$(\lambda,\mu)$ is clear from context.

Using the scalar multiplication, we introduce a left division, which
will allows us to characterize telescopers and represent certificates
in a compact form. To this end, we need to define the notion of orders
in $\lmA$. Let $M = \sum_{i=m}^na_i\lmshift^i\in\lmA$, where
$m,n\in\set Z$ with $m\leq n$ and $a_i\in\set K(x,y)$ with $a_ma_n\neq
0$.  We say that $m$ and $n$ are the lowest and highest orders of $M$,
and denote them by $\lord(M)$ and $\hord(M)$, respectively.
\begin{lemma}\label{LEM:order}
  Let $L\in\set K(x,y)[\emph\shift_y,\emph\shift_y^{-1}]$ and
  $M\in\lmA$ be two nonzero operators. Then $L\odot M$ is
  nonzero. Moreover,
  \[
  \lord(L \odot M) = \lord(\lmphi(L)) + \lord(M) 
  \quad \text{and} \quad 
  \hord(L \odot M) = \hord(\lmphi(L)) + \hord(M).
  \]
\end{lemma}
\begin{proof}
  Let 
  \begin{equation}\label{EQ:ops}
    L = \sum_{i=k}^\ell a_i \shift_y^i
    \in\set K(x,y)[\emph\shift_y,\emph\shift_y^{-1}]
    \quad \text{and}\quad 
    M = \sum_{j=m}^n b_j\lmshift^j \in\lmA,
  \end{equation}
  where $k,\ell,m,n\in\set Z$ with $k\leq \ell$ and $m\leq n$, and
  $a_i,b_j\in\set K(x,y)$ with $a_ka_\ell b_mb_n\neq 0$. Then
  $\lmphi(L) = \sum_{i=k}^\ell a_i \lmshift^{i\mu}$, which is nonzero.
  Hence, $\lord(\lmphi(L)) = k\mu$ and $\hord(\lmphi(L)) = \ell\mu$.
  Observe that $\shift_y^i \odot (f\lmshift^j) =
  \sigma_y^i(f)\lmshift^{i\mu+j}$ for all $f\in\set K(x,y)$ and
  $i,j\in\set Z$. Since $a_ka_\ell b_mb_n\neq 0$, then $\lord(L\odot
  M) = k\mu+m$ and $\hord(L\odot M) = \ell\mu+n$. In particular, $L
  \odot M\neq 0$.
\end{proof}
\begin{lemma}\label{LEM:leftquorem}
  Let $L\in\set K(x,y)[\emph\shift_y]$ with $L\neq 0$ and $M\in \lmA$.
  Then there exist $Q,R\in \lmA$ such that $M = L\odot Q + R$, and $R$
  either is zero or satisfies
  $0\leq\lord(R)\leq\hord(R)<\hord(\lmphi(L))$.
\end{lemma}
\begin{proof}
  If $M = 0$, then we set $Q = 0$ and $R = 0$. Otherwise, let $L$ and
  $M$ be the same as in \eqref{EQ:ops} with $k\geq 0$. Then
  $\hord(\lmphi(L))=\ell\mu$.
  
  \smallskip\noindent{\em Case 1.} Assume that $m\geq 0$. If
  $n<\ell\mu$, then we set $Q = 0$ and $R = M$. Otherwise, let $f =
  \sigma_y^{-\ell}(b_n/a_\ell)$. By Lemma~\ref{LEM:order},
  \[
  N := M - L \odot (f \lmshift^{n-\ell\mu}) 
  = M - (b_n \lmshift^n + \text{lower terms in $\lmshift$}). 
  \]
  Therefore, either $N = 0$ or $0\leq \lord(N) \leq \hord(N)<n$. If $N = 0$
  or $\hord(N) < \ell\mu$, then we are done. Otherwise, we recursively
  apply the same reduction on $N$. The conclusion will be reached in a
  finite number of steps.
	
  \smallskip\noindent{\em Case 2.} Assume that $m < 0$. We reduce $M$
  to an integer-linear operator which is either zero or of nonnegative
  lowest order. Let $g = \sigma_y^{-k}(b_m/a_k)$. Again, by
  Lemma~\ref{LEM:order}, $M - L \odot (g \lmshift^{m-k\mu}) = M - (b_m
  \lmshift^m + \text{higher terms in $\lmshift$})$, which is either
  zero or of lowest order higher than $m$. Repeating the above
  reduction finitely many times, we will obtain $Q_1, R_1\in\lmA$ such
  that $M = L\odot Q_1+R_1$ and either $R_1=0$ or $\lord(R_1)\geq 0$.
  If $R_1 = 0$, then we are done. Otherwise, applying the argument in
  the first case to $R_1$ yields the lemma.
\end{proof}

\begin{theorem}\label{THM:leftquorem}
  Let $L\in \set K(x,y)[\emph\shift_y,\emph\shift_y^{-1}]$ with $L\neq
  0$ and $M\in \lmA$. Then there exist unique $Q,R\in \lmA$ such that
  $M = L\odot Q + R$, and $R$ either is zero or satisfies
  \[
  \lord(\lmphi(L))\leq\lord(R)\leq\hord(R) < \hord(\lmphi(L)).
  \]
\end{theorem}
\begin{proof}
  Let $L$ be given as in \eqref{EQ:ops}. If $k\geq 0$, then the
  existence of $Q$ and $R$ follows from Lemma~\ref{LEM:leftquorem}.
  Assume that $k < 0$. The same lemma implies that there exist $\tilde
  Q,\tilde R\in\lmA$ such that $\shift_y^{-k}\odot M =
  (\shift_y^{-k}L)\odot \tilde Q + \tilde R$.  In addition, either
  $\tilde R = 0$ or $0\leq\lord(\tilde R)\leq\hord(\tilde
  R)<\hord(\lmphi(\shift_y^{-k}L))$.  It follows that $M = L\odot
  \tilde Q + \shift_y^k\odot \tilde R$.  Assume that $\tilde R$ is
  nonzero. Then $\lord(\shift_y^k\odot \tilde R) \geq
  \lord(\lmphi(L))$ by Lemma~\ref{LEM:order} and the fact that
  $\lord(\tilde R) \geq 0$.  Moreover, $\hord(\shift_y^k\odot \tilde
  R) < \hord(\lmphi(L))$ by Lemma~\ref{LEM:order} and the fact that
  $\hord(\tilde R)<\hord(\lmphi(\shift_y^{-k}L))$. Setting $Q = \tilde
  Q$ and $R = \shift_y^k\odot \tilde R$ establishes the existence of
  $Q$ and $R$.
	
  To show the uniqueness, we let $\bar Q, \bar R\in \lmA$ be such that
  $M = L\odot \bar Q + \bar R$, and $\bar R$ either is zero or
  satisfies $\lord(\lmphi(L))\leq \lord(\bar R)\leq\hord(\bar
  R)<\hord(\lmphi(L))$.  Then $L\odot(Q-\bar Q) = \bar R-R$. Suppose
  that $Q \neq \bar Q$.  Then $\bar R \neq R$ by
  Lemma~\ref{LEM:order}. Suppose that $\hord(Q-\bar Q)\geq 0$. By
  Lemma~\ref{LEM:order} and the fact that $\hord(\lmphi(L)) >
  \hord(\bar R-R)$, we have $\hord(L\odot(Q-\bar Q)) > \hord(\bar
  R-R)$, a contradiction.  Otherwise, a similar argument yields
  $\lord(L\odot (Q-\bar Q)) < \lord(\bar R-R)$, a contradiction.
  Hence $Q = \bar Q$ and then $R = \bar R$.
\end{proof}
In view	of the above theorem, we call $Q$ the {\em left scalar quotient}
and $R$ the {\em left scalar remainder} of $M$ by $L$, and denote them
by $\leftquo(M,L)$ and $\leftrem(M,L)$, respectively.
\begin{remark}\label{REM:leftquorem}
  It is possible to extend Theorem~\ref{THM:leftquorem} to the general
  case when the scalar divisor $L$ is an arbitrary nonzero operator in
  $\cal A$. However, as doing this extension is somewhat tedious and
  as this extension is not used in the paper we do not investigate
  this aspect further.
\end{remark}
\begin{remark}\label{REM:formulas}
  We are particularly interested in the case where the difference
  operator $\emph\shift_y-1$ plays the part of a scalar divisor.  For
  later reference, we collect below explicit formulas for left scalar
  remainders, as well as for left scalar quotients, in this case.
  
  Let $M =\sum_{i = m}^n a_i\emph\shift_{\lambda,\mu}^i\in\lmA$, where
  $m,n\in\set Z$ with $m\leq n$ and $a_i\in \set K(x,y)$.  Then
  \begin{equation}\label{EQ:leftremformula}
    \leftrem(M,\emph\shift_y-1)=\sum_{r=0}^{\mu-1}
    \left(\sum_{i_r}\sigma_y^{-q_{i_r}}(a_{i_r})\right)\lmshift^r,
  \end{equation}
  where the inner summation runs over all integers $i_r$ with
  $m\leq i_r\leq n$ such that $i_r = \mu q_{i_r}+r$ for some
  integer $q_{i_r}$, and
  \[
  \leftquo(M,\emph\shift_y-1) 
  = -\sum_{j= m}^{-1}\left(\sum_{i_j}\sigma_y^{-q_{i_j}}(a_{i_j})\right)
  \emph\shift_{\lambda,\mu}^j
  +\sum_{j=0}^{n-\mu}\left(\sum_{i_j} \sigma_y^{-q_{i_j}}(a_{i_j})\right)
  \emph\shift_{\lambda,\mu}^j,
  \]
  where the first inner summation runs over all integers $i_j$ with
  $m\leq i_j\leq n$ such that $i_j=\mu q_{i_j}+j$ for some nonpositive
  integer $q_{i_j}$, while the second inner summation runs over all
  integers $i_j$ with $m\leq i_j\leq n$ such that $i_j=\mu q_{i_j}+j$
  for some positive integer $q_{i_j}$.
\end{remark}
\begin{example}\label{EX:leftquorem}
  Let $M = (x-1)\emph\shift_{-5,2}+(2x^2+1)$ with
  $\emph\shift_{-5,2}=\emph\shift_x\emph\shift_y^3$. Let $L$ be an
  operator in $\set Q[x][\emph\shift_x]$ of the form
  $L = c_2\emph\shift_x^2+c_1\emph\shift_x+c_0$ for some
  $c_0,c_1,c_2\in \set Q[x]$. Multiplying $L$ from the left-hand side
  to $M$ yields
  \begin{align}
    L\odot M& = c_0 (x-1) \emph\shift_{-5,2} + c_0\, (2x^2+1)
    + c_1\, \sigma_x(x-1)\emph\shift_{-5,2}^{-4}
    + c_1\, \sigma_x(2x^2+1)\emph\shift_{-5,2}^{-5}
    \nonumber\\
    &\quad
    + c_2\,\sigma_x^2(x-1)\emph\shift_{-5,2}^{-9}
    + c_2\,\sigma_x^2(2x^2+1)\emph\shift_{-5,2}^{-10}.
    \label{EQ:sparse}
  \end{align}
  A direct calculation based on Remark~\ref{REM:formulas} then
  delivers
  \begin{equation}\label{EQ:leftrem}
    \leftrem(L\odot M, \emph\shift_y-1) 
    = a_1\emph\shift_{-5,2}+a_0,
  \end{equation}
  where 
  \begin{align*}
    &a_1 = c_0(x-1)+c_1(\sigma_y^3\sigma_x(2x^2+1))
    +c_2(\sigma_y^5\sigma_x^2(x-1))\\
    \text{and}\quad 
    &a_0 = c_0(2x^2+1)+c_1(\sigma_y^2\sigma_x(x-1))
    +c_2(\sigma_y^5\sigma_x^2(2x^2+1)).
  \end{align*}
  We note that $L\odot M$ is a sparse operator by \eqref{EQ:sparse};
  the left scalar quotient $\leftquo(L\odot M,\emph\shift_y-1)$,
  however, is a dense operator with exponents in $\emph\shift_{-5,2}$
  ranging consecutively from $-10$ to $-1$.
\end{example}

\section{Telescoping with compact certificates}\label{SEC:oct}
In this section, we demonstrate how to construct a telescoper for a
given rational function, along with its certificate in a compact form,
using left scalar divisions of integer-linear operators introduced in
the preceding section.

For an operator $L = \sum_{i,j\in\set Z}a_{ij}\shift_x^i\shift_y^j\in\cal A$
and a rational function $f\in \set K(x,y)$, the application of $L$ to
$f$ is defined to be
\[
L(f) = \sum_{i,j\in\set Z}a_{ij}\sigma_x^i(f)\sigma_y^j(f).
\]
Let $(\lambda,\mu)$ be a pair of coprime integers with $\mu>0$, and
$g\in \set K(x,y)$ of the form $g = p(\lambda x+ \mu y)$ with
$p\in\set K(z)$.  Then
\[
\lmshift^i(g) = p(\lambda x+ \mu y+i) \quad\text{for all}\ i\in\set Z.
\]
It follows that $\shift_x(g) = \lmshift^\lambda(g)$ and $\shift_y(g) =
\lmshift^\mu(g)$. Thus, for all $L\in\cal A$, we have that
$L(g)=\lmphi(L)(g)$. Assume further that $M\in \lmA$. Then
\begin{equation} \label{EQ:app} 
  LM(g) =(L \odot M)(g), 
\end{equation}
which allows us to describe telescopers and their corresponding
certificates in terms of module-theoretic language.

Let $f\in \set K(x,y)$ be a rational function with denominator
$g\in\set K[x,y]$.  Based on the refined integer-linear decomposition
\eqref{EQ:refinedildecomp} of $g$, there is a unique partial fraction
decomposition of $f$ with respect to~$y$, that is, there exist unique
$a_0,a_{ijk}\in\set K(x)[y]$ with $\deg_y(a_{ijk})<\deg_z(p_i)$ such
that
\begin{equation}\label{EQ:ratpfd}
  f= \frac{a_0}{p_0} +
  \sum_{i=1}^m\sum_{j=1}^{n_i}\sum_{k=1}^{e_{ij}}
  \frac{a_{ijk}}{p_i(\lambda_i x+\mu_i y+\nu_{ij})^k}.
\end{equation}
Let $d_i = \max_{1\leq j\leq n_i}\{e_{ij}\}$ and specify that $a_{ijk}
= 0$ in case $k>e_{ij}$. Interchanging the order of summations in
\eqref{EQ:ratpfd} and introducing the operator $M_{ik} =
\sum_{j=1}^{n_i} a_{ijk} \lmshift[i]^{\nu_{ij}}$ then gives
\begin{equation}\label{EQ:ratdecomp}
  f = \frac{a_0}{p_0} + \sum_{i=1}^m\sum_{k=1}^{d_i}
  M_{ik}\left(\frac1{p_i(\lambda_ix+\mu_iy)^k}\right).
\end{equation}
Note that $M_{ik}\in \set K(x)[y,\lmshift[i]]$ and
$\deg_y(M_{ik})<\deg_z(p_i)$ for all $i = 1,\dots,m$ and $k = 1,
\dots,d_i$. Using the above argument in the opposite direction, one
can easily derive the partial fraction decomposition \eqref{EQ:ratpfd}
from \eqref{EQ:ratdecomp}.  It thus follows from the uniqueness of
\eqref{EQ:ratpfd} that \eqref{EQ:ratdecomp} is unique. In particular,
the operators $M_{ik}$ are uniquely determined by the given rational
function $f$.  We will refer to \eqref{EQ:ratdecomp} as the
{\em RILD-based partial fraction decomposition} of $f$.

\subsection{The basic case}
In order to illustrate the main idea of our algorithm in a concise
way, we first focus on the simpler yet important subcase when the
given rational function $f\in \set K(x,y)$ possesses the form
\begin{equation}\label{EQ:basiccase}
  f = \sum_{i\in\set Z}\frac{a_i}{p(\lambda x+ \mu y + i)^k}
  = M\left(\frac1{p(\lambda x+ \mu y)^k}\right),
\end{equation}
where $\lambda,\mu,k\in\set Z$ with $\gcd(\lambda,\mu) = 1$ and
$\mu,k>0$, $p\in \set K[z]$ is monic and irreducible, $a_i\in \set
K(x)[y]$, finitely many nonzero, with $\deg_y(a_i) < \deg_z(p)$, and
$M=\sum_{i\in\set Z} a_i\lmshift^i \in \set
K(x)[y,\lmshift,\lmshift^{-1}]$.  Note that such a function $f$ has a
telescoper by the criterion of~\citet[Theorem~1]{AbLe2002}.

\begin{proposition}\label{PROP:basiccriterion}
  Let $f\in \set K(x,y)$ be of the form \eqref{EQ:basiccase}, and let
  $L\in\set K[x][\emph\shift_x]$ be a nonzero operator. Then $L$ is a
  telescoper for $f$ if and only if there exists $Q\in\lmA$ such that
  $L\odot M = (\emph\shift_y-1)\odot Q$.  When this is the case, a
  certificate corresponding to $L$ is given by
  \[
  Q\left(\frac1{p(\lambda x+\mu y)^k}\right).
  \]
\end{proposition}
\begin{proof}
  Assume that $L$ is a telescoper for $f$. Then there exists 
  $h\in\set K(x,y)$ such that 
  \[
  L(f) = (L\odot M) \left(\frac1{p(\lambda x + \mu y)^k}\right)
  = (\shift_y-1)(h),
  \]
  where the first equality follows by \eqref{EQ:app}. From the 
  uniqueness of the RILD-based partial fraction decomposition we have 
  that there exists $Q\in \set K(x)[y,\lmshift,\lmshift^{-1}]$ with
  $\deg_y(Q) < \deg_z(p)$ such that
  \[
  h = Q\left(\frac1{p(\lambda x + \mu y)^k}\right).
  \] 
  It thus follows that
  \[
  (L\odot M) \left(\frac1{p(\lambda x + \mu y)^k}\right) 
  = \big((\shift_y-1)\odot Q\big)\left(\frac1{p(\lambda x + \mu y)^k}\right).
  \]
  Again, by the uniqueness of the RILD-based partial fraction
  decomposition, we find that $L \odot M = (\shift_y-1)\odot
  Q$. Applying both sides of this equality to $1/p(\lambda x + \mu
  y)^k$ proves the converse.
\end{proof}
Based on Proposition~\ref{PROP:basiccriterion}, for a rational
function $f\in \set K(x,y)$ of the form \eqref{EQ:basiccase}, we can
compute a telescoper of order no more than $\rho \in\set N$ as
follows. Making an ansatz $L = c_\rho \shift_x^\rho + \dots + c_0$
with $c_0,\dots,c_\rho\in\set K[x]$ to be determined, we first compute
the left scalar remainder $R$ of $L\odot M$ by $\shift_y-1$.  Note
that $R\in\set K(x)[y,\lmshift,\lmshift^{-1}]$.  Sending $R$ to zero
thus yields a linear homogeneous system in $c_0,\dots,c_\rho$ over
$\set K(x)$.  If this system admits a nontrivial solution over $\set
K[x]$, we then successfully find a desired telescoper.  Otherwise, we
have shown that such a telescoper does not exist. Performing the above
procedure for $\rho = 0,1,\dots$, one eventually obtains a minimal
telescoper for $f$. With a telescoper $L$ for the given rational
function $f$ at hand, by Proposition~\ref{PROP:basiccriterion}, a
corresponding certificate can be formally represented as
$\leftquo(L\odot M, \shift_y-1)\left(\frac1{p(\lambda x+\mu y)^k}\right)$.
\begin{example}\label{EX:octsmall}
  Consider the rational function $f$ of the form
  \begin{equation*}\label{EQ:r1pfd}
    f = \frac{2x^2+1}{(-5x+2y)^2+1}+\frac{x-1}{(-5x+2y+1)^2+1}.
  \end{equation*}
  A simple calculation shows that the RILD-based partial fraction
  decomposition of $f$ is given by
  \[
  f = \underbrace{((x-1)\emph\shift_{-5,2}+2x^2+1)}_{M}
  \left(\frac{1}{p(-5x+2y)}\right)
  \quad\text{with}\
  \emph\shift_{-5,2}=\emph\shift_x\emph\shift_y^3
  \ \text{and}\ p(z) = z^2+1.
  \]
  Let $L = c_2\emph\shift_x^2+c_1\emph\shift_x+c_0$ with
  $c_0,c_1,c_2\in \set Q[x]$ to be determined. By
  Example~\ref{EX:leftquorem}, we know that the left scalar remainder
  $R$ of $L\odot M$ by $\emph\shift_y-1$ is given by
  \eqref{EQ:leftrem}.  Sending $R$ to zero then delivers the following
  linear homogeneous system
  \begin{equation}\label{EQ:r1system}
    \begin{pmatrix}
      x-1 & \sigma_y^3\sigma_x(2x^2+1) & \sigma_y^5\sigma_x^2(x-1)\\[2ex]
      2x^2+1 & \sigma_y^2\sigma_x(x-1) & \sigma_y^5\sigma_x^2(2x^2+1)
    \end{pmatrix}
    \begin{pmatrix}
      c_0\\c_1\\c_2
    \end{pmatrix}
    =\begin{pmatrix} 0\\[2ex] 0\end{pmatrix}.
  \end{equation}
  Solving this system over $\set Q[x]$ gives a telescoper
  \begin{equation*}\label{EQ:octsmallL1}
    L =(4x^4+8x^3+7x^2+5x+3)\emph\shift_x^2+2(2x^2-5)\emph\shift_x
    -(4x^4+24x^3+55x^2+59x+27)
  \end{equation*}
  and then a corresponding certificate 
  \begin{equation*}
    h=\leftquo(L\odot M,\emph\shift_y-1)\left(\frac1{p(-5x+2y)}\right).
  \end{equation*}
  We note that $L$ is actually a telescoper for $f$ of minimal order.
\end{example}

\subsection{The general case}
We now turn our attention to the general case, namely the case when
the input is an arbitrary rational function in $\set K(x,y)$.  Let
$f\in \set K(x,y)$ be a rational function admitting the RILD-based
partial fraction decomposition \eqref{EQ:ratdecomp}.  By
\cite[Theorem~1]{AbLe2002}, $f$ has a telescoper if and only if
$a_0/p_0$ in \eqref{EQ:ratdecomp} is $\sigma_y$-summable. Thus it
suffices to construct a telescoper for $r := f-a_0/p_0$, which
possesses the following form
\begin{equation}\label{EQ:rform}
  r = \sum_{i=1}^m\sum_{k=1}^{d_i} M_{ik}
  \left(\frac1{p_i(\lambda_ix+\mu_iy)^k}\right),
\end{equation}
where each $(\lambda_i,\mu_i)$ is a pair of coprime integers with
$\mu_i>0$, each $p_i\in\set K[z]$ is monic and irreducible, each 
$M_{ik}\in \set K(x)[y,\lmshift[i]]$ with $\deg_y(M_{ik})<\deg_z(p_i)$,
and the $p_i(\lambda_i x+ \mu_iy)$ are pairwise
$(\sigma_x,\sigma_y)$-inequivalent

There are two natural ways to proceed. The first method separately
takes each simple fraction
$M_{ik}\left(\frac1{p_i(\lambda_ix+\mu_iy)^k}\right)$ in
\eqref{EQ:rform} as the basic case and computes its own minimal
telescoper $L_{ik}\in\set K[x][\shift_x]$ using the approach presented
in the preceding subsection, and then returns the least common left
multiple of all these $L_{ik}$ as the output. By taking use of
\cite[Theorem~2]{Le2003a}, one can show that this least common left
multiple gives a minimal telescoper for $r$ (and thus for $f$).
Preliminary experiments, however, suggest that in practice this method
does not perform as well as expected. In fact, it is often less
efficient than the second method which we are going to explore
shortly.

This second method shares exactly the same spirit as the basic case
given in the preceding subsection, in the sense that it also reduces
the problem of constructing a telescoper to the problem of computing
left scalar remainders of integer-linear operators.
\begin{theorem}\label{THM:criterion}
  Let $r\in\set K(x,y)$ be a rational function possessing the
  decomposition \eqref{EQ:rform}, and let $L\in\set K[x][\emph\shift_y]$
  be a nonzero operator.  Then $L$ is a telescoper for $r$ if and only
  if, for all $i = 1,\dots,m$ and $k = 1,\dots,d_i$, there exist
  operators $Q_{ik}\in\lmA[i]$ such that
  $L\odot M_{ik} = (\emph\shift_y-1)\odot Q_{ik}$. When this is the
  case, a corresponding certificate of $L$ is given by
  \begin{equation}\label{EQ:certform}
    \sum_{i=1}^m\sum_{k=1}^{d_i}
    Q_{ik}\left(\frac1{p_i(\lambda_ix+\mu_iy)^k}\right).
  \end{equation}
\end{theorem}
\begin{proof}
  Assume that $L$ is a telescoper for $r$. Then there exists $h\in\set
  K(x,y)$ such that $L(r) = (\shift_y-1)(h)$. By \eqref{EQ:app} and
  \eqref{EQ:rform}, we have
  \[
  \sum_{i=1}^m\sum_{k=1}^{d_i}(L\odot M_{ik})
  \left(\frac1{p_i(\lambda_i x + \mu_i y)^k}\right)
  = (\shift_y-1)(h).
  \]
  The RILD-based partial fraction decomposition of the left-hand side
  in the above equality implies that the same decomposition of $h$ is
  of the form \eqref{EQ:certform}, in which $Q_{ik} \in \set
  K(x)[y,\lmshift[i],\lmshift[i]^{-1}]$ with $\deg_y(Q_{ik}) <
  \deg_z(p_i)$. The uniqueness of the RILD-based partial fraction
  decomposition then forces
  \[
  L\odot M_{ik} = (\shift_y-1)\odot Q_{ik}
  \quad\text{for all}\ i = 1,\dots,m \ \text{and}\ k = 1,\dots,d_i.
  \]
  
  Conversely, we apply $L$ to $r$. By \eqref{EQ:app} and
  \eqref{EQ:rform},
  \[
  L(r) = \sum_{i=1}^m\sum_{k=1}^{d_i} (L\odot M_{ik})
  \left(\frac1{p_i(\lambda_i x + \mu_i y)^k}\right)
  = \sum_{i=1}^m\sum_{k=1}^{d_i}\big( (\shift_y-1)\odot Q_{ik}\big)
  \left(\frac1{p_i(\lambda_i x + \mu_i y)^k}\right).
  \]
  It follows that $L(r) = (\shift_y-1)(h)$, where $h$ is given by the
  formula \eqref{EQ:certform}.
\end{proof}
In analogy to the basic case, the above theorem induces an iterative
strategy to compute a telescoper for a given rational function.

Putting this all together, we obtain a new creative telescoping
algorithm for rational functions.

\smallskip\noindent{\bf RationalCT.} Given a rational function $f\in
\set K(x,y)$, compute a minimal telescoper $L\in \set K[x][\shift_x]$
for $f$ and a corresponding certificate $h\in \set K(x,y)$ if
telescopers exist. The steps are:

\begin{enumerate}
\item Compute the RILD-based partial fraction decomposition of $f$ to
  get \eqref{EQ:ratdecomp}.
  
\item Apply the GGSZ reduction to $a_0/p_0$ in \eqref{EQ:ratdecomp} to
  find $h,r\in \set K(x,y)$ with $h$ being of a compact form such that
  \begin{equation}\label{EQ:nonintlinear}
    \frac{a_0}{p_0} = (\shift_y-1)(h)+r.
  \end{equation}
  
\item If $r\neq 0$ then return \lq\lq No telescoper exists!\rq\rq.
  
\item For $i=1,\dots,m$ and $k=1,\dots,d_i$ set $R_{ik} = 0$.\\[1ex]
  For $\ell = 0, 1, 2, \dots $ do
  \begin{itemize}
  \item[4.1] For $i=1,\dots,m$ and $k=1,\dots,d_i$ do
    \begin{itemize}
    \item[4.1.1] Compute the left scalar remainder $\bar R$ of 
      $\shift_x^\ell\odot M_{ik}$ by $\shift_y-1$.
      
    \item[4.1.2] Update $R_{ik}$ to be $R_{ik}+c_\ell\bar R$,
      where $c_\ell$ is an indeterminate.
    \end{itemize}
    
  \item[4.2] Find $c_0,\dots,c_\ell\in \set K[x]$ such that $R_{ik} = 0$
    for all $i=1,\dots,m$ and $k=1,\dots,d_i$, by solving a linear
    system in $c_0,\dots,c_\ell$ over $\set K[x]$. If there is a
    nontrivial solution, set $L = \sum_{j=0}^\ell c_j\shift_x^j$ and
    return
    \[
    \left(L,L(h)+
    \sum_{i=1}^m\sum_{k=1}^{d_i}\leftquo(L\odot M_{ik},\shift_y-1)
    \left(\frac1{p_i(\lambda_i x+\mu_i y)^k}\right)\right).
    \]
  \end{itemize}
\end{enumerate}
\begin{theorem}\label{THM:octcorrectness}
  Let $f$ be a rational function in $\set K(x,y)$. Then the algorithm
  {\bf RationalCT} terminates and correctly finds a minimal telescoper
  for $f$ and a corresponding certificate in a compact form when such
  telescopers exist.
\end{theorem}
\begin{proof}
  By \citep[Theorem~1]{AbLe2002}, $f$ has a telescoper if and only if
  $a_0/p_0$ in \eqref{EQ:ratdecomp} is $\sigma_y$-summable, which,
  according to \citep[Theorem~12]{GGSZ2003}, is equivalent to the
  condition that $r = 0$ in \eqref{EQ:nonintlinear}. Thus steps~1-3
  are correct.
	
  For $\ell = 0$, It is evident that $R_{ik}$ obtained in step~4.1 is
  equal to $\leftrem(c_0\odot M_{ik},\shift_y-1)$ for all $i =
  1,\dots,m$ and $k = 1,\dots,d_i$. By a direct induction on $\ell$,
  we see that in the outer loop of step~4, $R_{ik} =
  \leftrem((c_\ell\shift_x^\ell+\dots+c_0)\odot M_{ik}, \shift_y-1)$
  holds for all $i = 1,\dots,m$ and $k = 1,\dots,d_i$ every time the
  algorithm passes through step~4.1.
	
  Assume that $L = \sum_{\ell=0}^{\rho}\tilde c_\ell\shift_x^\ell$
  with $\tilde c_\ell\in\set K[x]$ and $\tilde c_{\rho}\neq 0$ is a
  minimal telescoper for $f$. By Theorem~\ref{THM:criterion}, the left
  scalar remainders $\tilde R_{ik}$ of the $L\odot M_{ik}$ by
  $\shift_y-1$ are all zero. Thus, the linear homogeneous system over
  $\set K[x]$ obtained by equating all the $R_{ik}$ at the $\rho$th
  iteration of the outer loop in step~4 to zero has a nontrivial
  solution, which gives rise to a telescoper of minimal order. The
  compact representation for a corresponding certificate follows by
  Theorem~\ref{THM:criterion}.
\end{proof}
\begin{example}\label{EX:oct}
  Consider the same rational function $f$ as in
  Example~\ref{EX:ggszred}.  By Example~\ref{EX:intlinear}, the
  refined integer-linear decomposition of the denominator $g$ is given
  by \eqref{EQ:grefined}. Then in step~1, we obtain the RILD-based
  partial fraction decomposition
  \begin{align*}
    f &= \frac{a_0}{p_0}+
    \underbrace{((x-1)\emph\shift_{-5,2}+2x^2+1)}_{M_1}
    \left(\frac{1}{p_1(-5x+2y)}\right)
    +\underbrace{\tfrac1{30}(-3x^2-x+10)\emph\shift_{3,10}^0}_{M_2}
    \left(\frac1{p_2(3x+10y)}\right),\\
    &\quad
    +\underbrace{
      \tfrac1{30}(9x^3+30x^2y-3x^2+10xy-29x-100y+20)
      \emph\shift_{3,10}^0}_{M_3}
    \left(\frac1{p_3(3x+10y)}\right),
  \end{align*}
  where $\emph\shift_{-5,2}=\emph\shift_x\emph\shift_y^3$,
  $\emph\shift_{3,10}=\emph\shift_x^7\emph\shift_y^{-2}$ and
  \begin{equation}\label{EQ:a0p0}
    \frac{a_0}{p_0}
    = \frac{2x+3}{\sigma_y^{30}(g_0)}-\frac{2x+3}{\sigma_y^{29}(g_0)}
    -\frac1{\sigma_y(g_0)}+\frac1{g_0},
  \end{equation}
  to which subsequently applying the GGSZ reduction in step~2 yields
  \eqref{EQ:nonintlinear} with $h$ represented by the compact form
  given in \eqref{EQ:red0} and $r =0$. In step~4, we execute the outer
  loop for $\ell = 0,1,\dots,22$ and iteratively compute the left
  scalar remainder $R_i$ of $(c_{22}\emph\shift_x^{22}+\dots+c_0)\odot
  M_i$ by $\emph\shift_y-1$ for $i = 1,2,3$, where $c_0,\dots,c_{22}$
  are indeterminates. By equating $R_1, R_2, R_3$ to zero, we set up
  an overdetermined system of $32$ linear equations in unknowns
  $c_0,\dots,c_{22}$ over $\set Q[x]$, in which each linear equation
  is of degree in $x$ at most 3. Solving this linear system over $\set
  Q[x]$ gives a minimal telescoper
  \begin{align}
    L&=(3x^2+42x+82)\emph\shift_x^{22}-(3x^2+30x+10)\emph\shift_x^{20}
    -2(3x^2+72x+142)\emph\shift_x^{12}\nonumber\\
    &\quad
    +2(3x^2+60x+10)\emph\shift_x^{10}+(3x^2+102x+802)\emph\shift_x^{2}
    -(3x^2+90x+610)\label{EQ:octL},
  \end{align}
  along with a corresponding certificate in the compact expression
  \begin{align*}
    &\quad L(h)+\leftquo(L\odot M_1,\emph\shift_y-1)
    \left(\frac{1}{p_1(-5x+2y)}\right)\\
    &+\leftquo(L\odot M_2,\emph\shift_y-1)\left(\frac1{p_2(3x+10y)}\right)
    +\leftquo(L\odot M_3,\emph\shift_y-1)\left(\frac1{p_3(3x+10y)}\right).
  \end{align*}
\end{example}

\subsection{Efficiency considerations}\label{SUBSEC:modifications}
The efficiency of the algorithm {\bf RationalCT} can be enhanced by
incorporating the following modifications in the algorithm.

\smallskip\noindent
(i) {\em Modification in step 1.}

In step~1, we employ the shift-homogeneous decomposition to obtain the
refined integer-linear decomposition of the denominator of $f$, which
leads to the RILD-based partial fraction decomposition of $f$.  In
fact, the role of the shift-homogeneous decomposition can be played by
any shiftless decomposition introduced in \cite[Definition~1]{GGSZ2003}.
In particular, the coarsest shiftless decomposition, namely the one
which groups all irreducible factors $g_i$ having the same tuples
$(\nu_{i1},\dots,\nu_{in_i})$ and $(e_{i1},\dots,e_{in_i})$ in
\eqref{EQ:shiftlessdecomp}, can be used.  Such a decomposition can be
computed via GCD computation (see \cite[\S 3]{GGSZ2003}). In this way,
we avoid the need of full factorization while maintaining the
uniqueness of the induced RILD-based partial fraction decomposition,
which in turn ensures the correctness of the algorithm.

\smallskip\noindent
(ii) {\em Modification in step~2.}  

In step~2, with $\tilde a_0\in \set K[x,y]$ and $u_0\in \set K[x]$
denoting the numerator and denominator of $a_0$, respectively, it
actually suffices to apply the GGSZ reduction to $\tilde a_0/p_0$
(instead of $a_0/p_0$) since $\tilde a_0/p_0 =
(\shift_y-1)(hu_0)+ru_0$ and $ru_0 = 0$ if and only if $r=0$. This
reduces the cost of this step.

\smallskip\noindent
(iii) {\em Modification in step~4.}  

Let $u\in\set K[x]$ be the common denominator of the $M_{ik}$ and
write each $M_{ik}$ as $M_{ik} = \frac1{u} \tilde M_{ik}$ for some
$\tilde M_{ik}\in\set K[x,y,\lmshift[i]]$.  Inspired by the proof of
\citep[Theorem~10]{ChKa2012a}, it actually amounts to looking for a
telescoper of the form
$L = \sum_{\ell=0}^\rho c_\ell\sigma_x^\ell(u)\shift_x^\ell$. As
such, for all $i = 1,\dots,m$ and $k = 1,\dots,d_i$, we have
\begin{equation}\label{EQ:cancelden}
  L\odot M_{ik}
  = \left(\sum_{\ell=0}^\rho c_\ell\sigma_x^\ell(u)\shift_x^\ell\frac1u\right)
  \odot \tilde M_{ik}
  = \left(\sum_{\ell=0}^\rho c_\ell\shift_x^\ell\right) \odot \tilde M_{ik}
  = \sum_{\ell=0}^\rho c_\ell\left(\shift_x^\ell\odot \tilde M_{ik}\right),
\end{equation}
and thus $L\odot M_{ik} \in \set K[x,y,\lmshift[i],\lmshift[i]^{-1}]$,
so that operations in step~4 only induce arithmetic with polynomial
coefficients. For doing so, we compute in step~4.1.1 the left scalar
remainder of $\tilde R$ of $\shift_x^\ell\odot \tilde M_{ik}$ (instead
of $\shift_x^\ell\odot M_{ik}$) by $\shift_y-1$ and return in step~4.2
a telescoper of the form
$L = \sum_{j=0}^\ell c_j\sigma_x^j(u)\shift_x^j$ once a
nontrivial solution is found.

We note that looking for a telescoper of the specified form in fact
does not lose any generality because, for a telescoper $\tilde L =
\sum_{\ell=0}^\rho\tilde c_\ell\shift_x^\ell\in \set K[x][\shift_x]$,
multiplying from the left the least common multiple $u_\rho$ of $u,
\sigma_x(u), \dots, \sigma_x^\rho(u)$ gives
\[
u_\rho \tilde L = \sum_{\ell=0}^\rho c_\ell\,\sigma_x^\ell(u)\shift_x^\ell
\quad\text{with}\
c_\ell = \frac{\tilde c_\ell u_\rho}{\sigma_x^\ell(u)}\in \set K[x],
\]
which is again a telescoper with the same order as $\tilde L$ and of
the required form. On the other hand, it is often observed in
experiments that taking such a special form for telescopers actually
helps to decrease sizes of the $c_\ell$ to be determined, which might
deserve further investigation.

\smallskip\noindent
(iv) {\em Further modification in step~4.}  

Following the preceding modification, before executing the outer loop
of step~4, we can first compute the left scalar remainder $N_{ik}$ of
each $\tilde M_{ik}$ by $\shift_y-1$; then we let $N_{ik}$ play the
role of $\tilde M_{ik}$ in step~4.1.1.  This is because any operator
in $\set K[x][\shift_x]$ commutes with $\shift_y-1$ and then
$\leftrem(\shift_x^\ell\odot \tilde M_{ik},\shift_y-1) =
\leftrem(\shift_x^\ell\odot N_{ik},\shift_y-1)$ for any $\ell\in \set
N$.  Note that every nonzero $N_{ik}$ has highest order no more than
$\mu_i-1$ and typically can be handled more easily than $\tilde
M_{ik}$.

Let us now reconsider Example~\ref{EX:oct} in the light of the above
modifications.
\begin{example}\label{EX:modified}
  Consider the same rational function $f$ as Example~\ref{EX:ggszred}.
  Using the coarsest shiftless decomposition instead of the
  shift-homogeneous decomposition as described in the modification~(i), we
  obtain the following refined integer-linear decomposition
  \begin{equation}\label{EQ:grild}
    g = p_0(x,y)p_1(-5x+2y) p_1(-5x+2y+1)p_2(3x+10y),
  \end{equation}
  where $p_0=g_0\sigma_y(g_0)\sigma_y^{29}(g_0)\sigma_y^{30}(g_0)$
  with $g_0 = xy +1$, $p_1(z) = z^2+1$ and $p_2(z) = z^3+1$.  Based on
  \eqref{EQ:grild}, we find in step~1 the partial fraction
  decomposition
  \begin{align*}
    f = \frac{a_0}{p_0}+ 
    \underbrace{((x-1)\emph\shift_{-5,2}+2x^2+1)}_{M_1}
    \left(\frac{1}{p_1(-5x+2y)}\right)
    +\underbrace{(xy+1)\emph\shift_{3,10}^0}_{M_2}
    \left(\frac1{p_2(3x+10y)}\right),
  \end{align*}
  where $a_0/p_0$ is given by \eqref{EQ:a0p0},
  $\emph\shift_{-5,2}=\emph\shift_x\emph\shift_y^3$ and
  $\emph\shift_{3,10}=\emph\shift_x^7\emph\shift_y^{-2}$.  Again, in
  step~2, we apply the GGSZ reduction to $a_0/p_0$ which yields
  \eqref{EQ:nonintlinear} with $h$ represented by the compact form
  given in \eqref{EQ:red0} and $r =0$.  In step~4, the loop will be
  executed for $\ell = 0,\dots,22$. The final, induced linear system
  contains $22$ equations in unknowns $c_0,\dots,c_{22}$ over $\set
  Q[x]$ and each equation has degree in $x$ at most 2. This compares
  to Example~\ref{EX:oct} which involves a linear system of $32$
  equations of degree in $x$ at most 3. The basis to the
  nullspace of the linear system over $\set Q(x)$ gives rise to the
  same minimal telescoper $L$ in \eqref{EQ:octL}.  Note that
  modifications~(ii)-(iv) are trivial in this example.
\end{example}

\section{Arithmetic cost for the new algorithm}\label{SEC:complexity}
In this section, we give a complexity analysis of the new algorithm
described in the preceding section. For this purpose, we first collect
some classical complexity notations and facts needed in this
paper. More background on these can be found in \citep{vzGe2013}.

\subsection{Complexity background}
In this paper, costs of algorithms will be counted by the
number of arithmetic operations in the field $\set K$. All costs are
analyzed in terms of $\bigO$-estimates for classical arithmetic and
$\softO$-estimates for fast arithmetic, where the {\em soft-Oh notation}
\lq\lq $\softO$\rq\rq\ is basically \lq\lq $\bigO$\rq\rq\ but
suppressing logarithmic factors (see \citep[Definition~25.8]{vzGe2013}
for a precise definition).

We summarize the facts needed for our analysis below and will freely
use them later.  For proofs, we refer to \citep{vzGe2013},
\citep[\S 3 and \S 5]{Gerh2004} and \citep[Theorem~4.1]{ZLS2012}.

The first fact gives sharp degree bounds for two basic arithmetic
operations -- division with remainder and partial fraction
decomposition.  This turns out to be very useful in estimating degree
sizes.  The proofs are mainly based on Cramer's rule and determinant
expansions and will be skipped.

\begin{fact}[Degree bounds]\label{FAC:degreesize}
  Let $f,g$ be two nonzero polynomials in $\set K[x,y]$.
  \begin{itemize}
  \item[(i)] Assume that $\deg_y(f)\geq \deg_y(g)$. Then there exist
    unique $q,r\in \set K[x,y]$ with
    \begin{align*}
      &(\deg_x(q),\deg_y(q)) \leq
      \big((\deg_y(f)-\deg_y(g))\deg_x(g)+\deg_x(f),
      \deg_y(f)-\deg_y(g)\big)\\[1ex]
      \text{and}\quad &(\deg_x(r),\deg_y(r)) \leq
      \big((\deg_y(f)-\deg_y(g)+1)\deg_x(g)+\deg_x(f), \deg_y(g)-1\big)
    \end{align*}
    such that $\lc_y(g)^{\deg_y(f)-\deg_y(g)+1} f = q g + r$.
		
  \item[(ii)] Assume that $\deg_y(f)<\deg_y(g)$ and $g= g_1^{e_1}\dots g_m^{e_m}$
    with $e_i\in\set N\setminus\{0\}$ and $g_i\in \set K[x,y]$ being
    pairwise coprime. Then there exists $u\in \set K[x]$ and
    $\{f_{ij}\}_{1\leq i\leq m,1\leq j\leq e_i}\subseteq \set K[x,y]$
    with
    \begin{align*}
      &\deg_x(u)\leq \deg_x(g)\deg_y(g)
      -\sum_{i=1}^m\frac{e_i(1+e_i)}2\deg_x(g_i)\deg_y(g_i)
      \quad\text{and}\\[1ex]
      &(\deg_x(f_{ij}),\deg_y(f_{ij}))\leq
      (\deg_x(g)\deg_y(g)
      +\deg_x(f)-\deg_x(g)+j\deg_x(g_i), \deg_y(g_i)-1)
    \end{align*}
    such that
    \[\frac{f}{g} = \frac1{u}\left(\frac{f_{11}}{g_1}+\dots+
    \frac{f_{1e_1}}{g_1^{e_1}}+\dots+\frac{f_{m1}}{g_m}+\dots
    +\frac{f_{me_m}}{g_m^{e_m}}\right).\]
  \end{itemize}
\end{fact}

The next fact contains the cost of some basic arithmetics for
univariate polynomials.
\begin{fact}[Arithmetic of univariate polynomials]
  \label{FAC:unipolycost}
  Let $f,g\in \set K[x]$ with $\deg_x(f),\deg_x(g)\leq d_x$. Then the
  following operations can be performed at most in $\bigO(d_x^2)$
  arithmetic operations in $\set K$ with classical arithmetic and
  $\softO(d_x)$ with fast arithmetic.
  \begin{itemize}
  \item[(i)] Addition, multiplication, division with remainder, GCD
    computation of $f$ and $g$;
  \item[(ii)] Evaluation $f$ at $d_x+1$ distinct points in $\set K$ or
    interpolation in $\set K[x]$ at these points;
  \item[(iii)] Partial fraction decomposition of $f/g$ with respect to
    a given factorization of $g$, provided that $f,g$ are nonzero
    coprime polynomials with $\deg_x(f)<\deg_x(g)$.
  \end{itemize}
\end{fact}
In order to analyze the cost for operations on bivariate polynomials,
a general (although not optimal) technique is to use evaluation and
interpolation on polynomials and to perform operations on univariate
polynomials based on the above fact. We will frequently use this
technique without explicitly pointing it out.

As mentioned in the introduction, most of recent creative telescoping
algorithms, including our new one presented in Section~\ref{SEC:oct},
eventually reduce the problem of finding telescopers to the problem of
solving linear systems, which can be accomplished efficiently.
\begin{fact}[Solving linear systems]\label{FAC:systemcost}
  Let $M$ be a polynomial matrix in $\set K[x]^{m\times n}$ with
  entries being polynomials in $\set K[x]$ of degree in $x$ less than
  $d_x$. Assume that $n\in\bigO(m)$. Then a basis of the null space of
  $M$ in $\set K[x]$ can be computed using $\bigO(m^3d_x^2)$
  arithmetic operations in~$\set K$ with classical arithmetic
  (Gaussian elimination) and $\softO(m^{\omega-1}nd_x)$ with fast
  arithmetic, where $\omega\in\set R$ with $2<\omega\leq 3$ is the
  exponent of matrix multiplication over $\set K$.
\end{fact}

\subsection{Output size estimates}
We define the degree of a rational function in $\set K(x,y)$ with
respect to $x$ (resp.\ $y$) to be the maximum of the degrees of its
numerator and denominator with respect to $x$ (resp.\ $y$).  Using
Fact~\ref{FAC:degreesize}, we are now able to estimate sizes of
intermediate results.
\begin{lemma}\label{LEM:intermediatedeg}
  Let $f\in\set K(x,y)$ be a rational function with $\deg_x(f) = d_x$
  and $\deg_y(f) = d_y$. Assume that the RILD-based partial fraction
  decomposition of $f$ takes the form \eqref{EQ:ratdecomp}. Let
  $\tilde a_0 \in \set K[x,y]$ be the numerator of $a_0$.  Let
  $u\in\set K[x]$ be the common denominator of the $M_{ik}$ and write
  each $M_{ik}$ as $M_{ik} = \frac1{u} \tilde M_{ik}$ for some $\tilde
  M_{ik}\in\set K[x,y,\emph\shift_{\lambda_i,\mu_i}]$.  Then
  \begin{align*}
    &(\deg_x(\tilde a_0),\deg_y(\tilde a_0))\in \bigO(d_xd_y)\times\bigO(d_y),
    \quad \deg_x(u)\in\bigO(d_xd_y)\\
    \text{and}\quad &(\deg_x(\tilde M_{ijk}),\deg_y(\tilde M_{ijk}))
    \in \bigO(d_xd_y)\times\bigO(\deg_z(p_i))
    \ \text{for all}\ i = 1, \dots, m \ \text{and}\ k = 1, \dots, d_i.
  \end{align*}
\end{lemma}
\begin{proof}
  We know from definition that \eqref{EQ:ratdecomp} gives the partial
  fraction decomposition of $f$ with respect to $y$, based on the
  refined integer-linear decomposition of its denominator. The degree
  bounds thus follow directly by Fact~\ref{FAC:degreesize}.
\end{proof}
\begin{lemma}\label{LEM:systemdeg}
  Let $r\in \set K(x,y)$ be a rational function of the form
  \eqref{EQ:rform}.  Let $u\in\set K[x]$ be the common denominator of
  the $M_{ik}$ and write each $M_{ik}$ as $M_{ik} = \frac1{u} \tilde
  M_{ik}$ for some $\tilde M_{ik}\in\set
  K[x,y,\emph\shift_{\lambda_i,\mu_i}]$.  Let $L = \sum_{\ell=0}^\rho
  c_\ell\sigma_x^\ell(u)\emph\shift_x^\ell \in\set
  K[x][\emph\shift_x]$ with $\rho\in\set N$ and $c_\ell\in\set K[x]$.
  Then for each integer pair $(i,k)$ with $1\leq i\leq m$ and $1\leq
  k\leq d_i$, the left scalar remainder $R_{ik}$ of $L\odot M_{ik}$ by
  $\emph\shift_y-1$ can be written as
  \begin{equation}\label{EQ:leftremform}
    R_{ik} = c_\rho \tilde R_{ik\rho} + \dots + c_0 \tilde R_{ik0},
  \end{equation}
  where $\tilde R_{ik\ell}\in\set K[x,y,\emph\shift_{\lambda_i,\mu_i}]$ with
  \[
  (\deg_x(\tilde R_{ik\ell}),\deg_y(\tilde R_{ik\ell}))
  \leq (\deg_x(\tilde M_{ik}),\deg_y(\tilde M_{ik}))
  \quad \text{and}\quad
  \deg_{x,y}(\tilde R_{ik\ell}) \leq \deg_{x,y}(\tilde M_{ik}).
  \]
  Here $\deg_{x,y}(\cdot)$ denotes the total degree of the argument
  with respect to $x,y$.
\end{lemma}
\begin{proof}
  For each integer pair $(i,k)$ with $1\leq i\leq m$ and $1\leq k\leq
  d_i$, it follows from \eqref{EQ:cancelden} that letting $\tilde
  R_{ik\ell} = \leftrem(\shift_x^\ell\odot \tilde M_{ik},\shift_y-1)$
  for all $\ell = 0,\dots,\rho$ gives the decomposition
  \eqref{EQ:leftremform}.  It remains to check the degree estimates of
  $\tilde R_{ik\ell}$, which in turn is an immediate result of
  \eqref{EQ:leftremformula}.
\end{proof}
The following depicts an order-degree curve of telescopers for
bivariate rational functions.
\begin{lemma}\label{LEM:octorderdegree}
  Let $r\in \set K(x,y)$ be a rational function of the form
  \eqref{EQ:rform}. Let $u\in\set K[x]$ be the common denominator of
  the $M_{ik}$ and write each $M_{ik}$ as $M_{ik} = \frac1{u} \tilde
  M_{ik}$ for some $\tilde M_{ik}\in\set
  K[x,y,\emph\shift_{\lambda_i,\mu_i}]$.  For each integer pair
  $(i,k)$ with $1\leq i\leq m$ and $1\leq k\leq d_i$, define
  $\alpha_{ik} = \max\{-1,\deg_{x,y}(\tilde M_{ik})\}$ and $\beta_{ik}
  = \max\{-1,\deg_y(\tilde M_{ik})\}$, and let
  \begin{equation}\label{EQ:upperbound}
    \rho_0=\sum_{i=1}^m\sum_{k=1}^{d_i}\mu_i(\beta_{ik}+1).
  \end{equation}
  Then for any nonnegative integer pair $(\rho,\tau)$ with $\rho \geq
  \rho_0$ and
  \begin{align}
    \tau &> \deg_x(u)-1+\frac{\sum_{i=1}^m\sum_{k=1}^{d_i}
      \mu_i(\alpha_{ik}-\frac12\beta_{ik})
      (\beta_{ik}+1)}{\rho+1-\rho_0},
    \label{EQ:degreebound}
  \end{align}
  there exists a telescoper for $r$ of order at most $\rho$ and degree
  at most $\tau$.
\end{lemma}
\begin{proof}
  Let $\rho,\tau\in \set N$ with $\rho \geq \rho_0$ and $\tau$
  satisfying \eqref{EQ:degreebound}. To prove the lemma, it is
  sufficient to show that there exist $c_0,\dots, c_\rho\in \set
  K[x]$, not all zero, with $\deg_x(c_\ell)\leq\tau-\deg_x(u)$ such
  that
  \begin{equation}\label{EQ:system}
    \leftrem\left(
    \Big(\sum_{\ell=0}^\rho c_\ell\sigma_x^\ell(u)\shift_x^\ell\Big)
    \odot M_{ik}, \shift_y-1\right) = 0
    \quad\text{for all}\ i = 1,\dots,m\ \text{and}\ k = 1,\dots, d_i,
  \end{equation}
  because then Theorem~\ref{THM:criterion} asserts that
  $\sum_{\ell=0}^\rho c_\ell\sigma_x^\ell(u)\shift_x^\ell$ gives a
  desired telescoper for $r$. Now we consider the linear system over
  $\set K$ (rather than $\set K[x]$) obtained by vanishing
  coefficients of like powers of $x$ and $y$ in \eqref{EQ:system}. In
  other words, we view the coefficients of the $c_\ell$ with respect
  to $x$, not the $c_\ell$ themselves, as unknowns. This then gives us
  $(\tau-\deg_x(u)+1)(\rho+1)$ unknowns in total.  On the other hand,
  we derive from Lemma~\ref{LEM:systemdeg} that each equation in
  \eqref{EQ:system} has total degree in $x,y$ at most $\tau -\deg_x(u)
  + \alpha_{ik}$ and degree in $y$ at most $\beta_{ik}$.  It follows
  that the induced linear system contains at most
  \[
  (\tau-\deg_x(u)+1)\rho_0+\sum_{i=1}^m\sum_{k=1}^{d_i}
  \mu_i(\alpha_{ik}-\frac12\beta_{ik})
  (\beta_{ik}+1)
  \]
  equations over $\set K$. Since $\rho \geq \rho_0$, one concludes
  from \eqref{EQ:degreebound} that the linear system over $\set K$
  resulting from \eqref{EQ:system} have more unknowns than equations,
  assuring such a nontrivial solution.
\end{proof}
We note that the left scalar remainders of the $\tilde M_{ik}$ by
$\shift_y-1$ can be employed to further refine the bounds given by
\eqref{EQ:upperbound} and \eqref{EQ:degreebound}.
\begin{remark}\label{REM:octorderdegree}
  Under the assumptions of the above lemma, in the context of
  \citep[\S 4]{ChKa2012a}, all $\tilde M_{ik}$ are actually in $\set
  K[x,\emph\shift_{\lambda_i,\mu_i}]$, yielding
  $\alpha_{ik}=\max\{-1,\deg_x(\tilde M_{ik})\}$ and $\beta_{ik}=0$.
  Then $\rho_0=\sum_{i=1}^m\sum_{k=1}^{d_i}\mu_i$ by
  \eqref{EQ:upperbound}, and \eqref{EQ:degreebound} becomes
  \[
  \tau  > \deg_x(u)-1
  +\frac{\sum_{i=1}^m\sum_{k=1}^{d_i}\mu_i\alpha_{ik}}{\rho+1-\rho_0},
  \]
  which coincides with the order-degree curve given in
  \citep[Theorem~10]{ChKa2012a} (after correcting the typos in the
  formula of the lower bound for $d$ there).
\end{remark}

\subsection{Cost analysis of algorithm}
Recall that the {\em auto-dispersion set} of a polynomial $g\in\set
K[x,y]$ with respect to $y$ consists of all integers $\ell$ such that
$\deg_y(\gcd(g,\sigma_y^\ell(g)))>0$.
\begin{lemma}\label{LEM:ggszred}
  Let $a_0,p_0\in \set K[x,y]$ be two coprime polynomials with
  $p_0\neq 0$ and $\deg_y(a_0/p_0) = d_y$. Then the GGSZ reduction
  computes $h,r\in \set K(x,y)$ with $h$ in a compact form such that
  \eqref{EQ:nonintlinear} holds, using
  $\bigO(\deg_x(p_0)d_y^4+\deg_x(p_0)^2d_y^3+\deg_x(a_0)\deg_x(p_0)d_y^2+\deg_x(a_0)^2d_y)$
  arithmetic operations in~$\set K$ with classical arithmetic and
  $\softO(\deg_x(p_0)d_y^3+\deg_x(a_0)d_y)$ with fast arithmetic, plus
  the cost of computing the auto-dispersion set of $p_0$ with respect
  to $y$.
\end{lemma}
\begin{proof}
  By \citep[Theorem~13]{GGSZ2003}, the cost of the GGSZ reduction is
  dominated by computing a shiftless decomposition of $p_0$ and the
  subsequent partial fraction decomposition of $a_0/p_0$.  By
  \citep[Theorem~10]{GGSZ2003} and making use of the
  evaluation-interpolation technique, one obtains that the former
  operation can be accomplished using
  $\bigO(\deg_x(p_0)d_y^4+\deg_x(p_0)^2d_y)$ arithmetic operations in
  $\set K$ with classical arithmetic and $\softO(\deg_x(p_0)d_y^3)$
  with fast arithmetic, plus the cost of computing the auto-dispersion
  set of $p_0$ with respect to $y$.  While the latter operation takes
  $\bigO(\deg_x(p_0)^2d_y^3+\deg_x(a_0)\deg_x(p_0)d_y^2+\deg_x(a_0)^2d_y)$
  with classical arithmetic and
  $\softO(\deg_x(p_0)d_y^2+\deg_x(a_0)d_y)$ with fast
  arithmetic. Combining these two costs concludes the lemma.
\end{proof}

Now we are ready to study the cost of the algorithm {\bf RationalCT},
in which we shall assume that the four enhancements discussed in
Section~\ref{SUBSEC:modifications} have been taken into account.
\begin{theorem}\label{THM:oct}
  Let $f\in \set K(x,y)$ be a rational function with $\deg_x(f) = d_x$
  and $\deg_y(f) = d_y$. Assume that $f$ has a telescoper and let
  $\rho$ be the actual order of its minimal telescopers. Further
  assume \eqref{EQ:ratdecomp} and \eqref{EQ:rform} hold, and define
  $\rho_0$ by \eqref{EQ:upperbound}.  Then the algorithm
  {\bf RationalCT} finds a minimal telescoper for $f$ and a
  certificate in a compact form using
  $\bigO(d_xd_y^4+\rho d_x^2d_y^3+\rho\rho_0^3 d_x^2d_y^2)$ arithmetic
  operations in~$\set K$ with classical arithmetic and
  $\softO(d_xd_y^3+\rho d_xd_y^2+ \rho^2\rho_0^{\omega-1} d_xd_y)$
  with fast arithmetic, plus the cost of computing auto-dispersion
  sets and finding rational roots.
\end{theorem}
\begin{proof}
  Based on the modification (i) in Section~\ref{SUBSEC:modifications}, in
  step~1, we incorporate the coarsest shiftless decomposition into the
  integer-linear decomposition to obtain the refined one of the
  denominator of $f$, which, by \citep[Theorem~3.5]{GHLZ2019} and
  \citep[Theorem~10]{GGSZ2003}, takes $\bigO(d_x^2d_y+d_xd_y^3+d_y^4)$
  arithmetic operations with classical arithmetic and
  $\softO(d_xd_y^2+d_y^3)$ with fast arithmetic, plus the cost of
  finding rational roots.  Therefore, the corresponding RILD-based
  integer-linear decomposition of $f$ can be obtained using
  $\bigO(d_x^2d_y^3+d_y^4)$ with classical arithmetic and
  $\softO(d_xd_y^2+d_y^3)$ with fast arithmetic in total. Regardless
  of the cost of computing auto-dispersion sets, one concludes from
  the modification (ii) and Lemmas~\ref{LEM:intermediatedeg},
  \ref{LEM:ggszred} that step~2 takes $\bigO(d_xd_y^4+d_x^2d_y^3)$
  with classical arithmetic and $\softO(d_xd_y^3)$ with fast
  arithmetic. By assumption, $r =0$ in \eqref{EQ:nonintlinear} and
  thus the algorithm continues after step~3.
  
  Based on modifications (iii)-(iv), we proceed to find the common
  denominator $u\in\set K[x]$ of the operators $M_{ik}$, reformulate
  each of them as $M_{ik} = \frac1{u} \tilde M_{ik}$ for $\tilde
  M_{ik}\in\set K[x,y,\lmshift[i]]$, and compute the left scalar
  remainders $N_{ik}$ of the $\tilde M_{ik}$ by $\shift_y-1$.  By
  Lemma~\ref{LEM:intermediatedeg}, $\deg_x(\tilde
  M_{ik})\in\bigO(d_xd_y)$ and $\deg_y(\tilde
  M_{ik})\in\bigO(\deg_z(p_i))$. It then follows from
  \eqref{EQ:leftremformula} that computing all the $N_{ik}$ in total
  requires $\bigO(d_x^2d_y^3)$ with classical arithmetic and
  $\softO(d_xd_y^2)$ with fast arithmetic.  Since $\deg_x(N_{ik}) \leq
  \deg_x(\tilde M_{ik})$ and $\deg_y(N_{ik})\leq \deg_y(\tilde
  M_{ik})$, for each iteration of the outer loop of step~4, the same
  cost applies to step~4.1 with $M_{ik}$ replaced by $N_{ik}$ as
  discussed in modifications (iii)-(iv).
	
  Since $\rho$ is the actual order of minimal telescopers for $f$, the
  outer loop of step~4 runs exactly $\rho$ iterations.  Thus the total
  cost of step~4.1 in the whole loop is $\bigO(\rho d_x^2d_y^3)$ with
  classical arithmetic and $\softO(\rho d_xd_y^2)$ with fast
  arithmetic.  For the $\ell$-th iteration with $0\leq \ell \leq
  \rho$, Lemmas~\ref{LEM:intermediatedeg} and \ref{LEM:systemdeg}
  assert that the coefficient matrix over $\set K[x]$ attached to the
  linear system obtained in step~4.2 has at most $\rho_0$ rows and
  $\ell+1$ columns, and each of its nonzero entries has degree in $x$
  in $\bigO(d_xd_y)$.  Thus Fact~\ref{FAC:systemcost} implies that
  finding a solution needs $\bigO(\rho_0^3d_x^2d_y^2)$ with classical
  arithmetic and $\softO(\ell\rho_0^{\omega-1} d_xd_y)$ with fast
  arithmetic.  This yields the total cost of
  $\bigO(\rho\rho_0^3d_x^2d_y^2)$ with classical arithmetic and
  $\softO(\rho^2 \rho_0^{\omega-1}d_xd_y)$ with fast arithmetic for
  solving linear systems in step~4.2 in the whole loop, as there are
  $\rho$ iterations.
	
  When a nontrivial solution of the linear system in step~4.2 is found, it virtually takes no
  arithmetic operations for returning the certificate in such a
  compact representation. By the modification (iii), we eventually
  construct a minimal telescoper of the form $L = \sum_{j=0}^\ell
  c_j\sigma_x^j(u)\shift_x^j$.  Computing the $\sigma_x^j(u)$ in the
  telescoper $L$ requires $\bigO(\rho d_x^2d_y^2)$ with classical
  arithmetic and $\softO(\rho d_xd_y)$ with fast arithmetic. In
  addition, by Lemma~\ref{LEM:octorderdegree},
  $\deg_x(c_j)\in\bigO(\rho_0d_xd_y)$.  Therefore, expanding the
  telescoper $L$ takes $\bigO(\rho_0d_x^2d_y^2)$ with classical
  arithmetic and $\softO(\rho_0d_xd_y)$ with fast arithmetic. The
  announced cost follows.
\end{proof}

\begin{corollary}\label{COR:oct}
  With the assumptions of Theorem~\ref{THM:oct}, further let
  $\mu=\max\{\mu_1,\dots,\mu_m\}$. Then $\rho_0\in\bigO(\mu d_y)$, and
  the algorithm {\bf RationalCT} takes $\bigO(\mu^4d_x^2d_y^6)$
  arithmetic operations in~$\set K$ with classical arithmetic and
  $\softO(\mu^{\omega+1} d_xd_y^{\omega+2})$ with fast arithmetic,
  plus the cost of computing auto-dispersion sets and finding rational
  roots.
\end{corollary}
\begin{proof}
  By assumption, with $u\in\set K[x]$ denoting the common denominator
  of the $M_{ik}$ in \eqref{EQ:ratdecomp}, each operator $M_{ik}$ has
  the form $M_{ik} = \frac1{u}\tilde M_{ik}$ for $\tilde M_{ik}\in\set
  K[x,y,\lmshift[i]]$ with $\deg_y(\tilde M_{ik}) < \deg_z(p_i)$. It
  follows from \eqref{EQ:upperbound} that $\rho_0\in\bigO(\mu d_y)$.
  Since $\rho$ is the actual order of minimal telescopers for $f$, we
  conclude from Lemma~\ref{LEM:octorderdegree} that $\rho\leq
  \rho_0$. The announced cost is then evident by
  Theorem~\ref{THM:oct}.
\end{proof}
\begin{remark}\label{REM:upperboundcost}
  Under the assumptions of the above corollary, according to
  Lemma~\ref{LEM:octorderdegree}, there exists a minimal telescoper
  for $f$ of total size in $\bigO(\mu^2d_xd_y^3)$.
\end{remark}
\begin{remark}\label{REM:intcost}
  In the case of $\set K = \set Q$, by incorporating the cost of
  computing the auto-dispersion set of an integer polynomial
  (cf.\ \citep[Theorem~14]{GGSZ2003}) and the cost of finding rational
  roots of an integer polynomial
  (cf.\ \citep[Theorem~15.21]{vzGe2013}), one sees from the above
  corollary that the algorithm {\bf RationalCT} has the total running
  time bounded by $(\mu+d_x+d_y+\log||f||_\infty)^{\bigO(1)}$ word
  operations, where the max-norm $||f||_\infty$ of $f \in \set Q(x,y)$
  is defined as the maximal absolute value of the integer coefficients
  appearing in the numerator and denominator of $f$ with respect to
  $x,y$. See \citep{Gerh2004,vzGe2013} for more information on word
  operations.
\end{remark}

\section{Arithmetic cost for the reduction-based approach}\label{SEC:rct}
In this section, we review the reduction-based creative telescoping
algorithm developed in \citep{CHKL2015} in the context of bivariate
rational functions and further analyze its cost in this setting. As
indicated by the name of the algorithm, a reduction method plays a
fundamental role. The original reduction method employed by
\citep{CHKL2015} in the rational case was developed by
\cite{Abra1975}. In order to highlight more significant discrepancies
between this creative telescoping algorithm and the one developed in
Section~\ref{SEC:oct}, we instead use the GGSZ reduction recalled in
Section~\ref{SUBSEC:ggszred} to carry out all the reduction steps in
the algorithm.

Before discussing the concrete algorithm, let us recall some notions.
As a generalization of auto-dispersion sets, the {\em dispersion set}
of a polynomial $f\in\set K[x,y]$ with respect to another polynomial
$g\in\set K[x,y]$ is defined to be the integer set
\[
\disset_y(f,g)=\{\ell\in\set Z\mid \deg_y(\gcd(f,\sigma_y^\ell(g)))>0\}.
\]
Such a dispersion set can be achieved by the algorithm of \cite{MaWr1994}
or by the procedure {\bf pDispersionSet} from \cite[\S 6]{GGSZ2003} in
the particular case where $\set K = \set Q$.

A polynomial in $\set K[x,y]$ is called {\em primitive} with respect
to $y$ (or $y$-primitive for short) if the greatest common divisor
over $\set K[x]$ of all its coefficients with respect to $y$ is equal
to one. A rational function in $\set K(x,y)$ is called {\em proper}
with respect to $y$ (or $y$-proper for short) if the degree of its
numerator with respect to $y$ is less than that of its denominator.
For a rational function $f\in\set K(x,y)$, another rational function
$r\in \set K(x,y)$ is called a {\em shift-remainder} with respect to
$y$ (or {\em $\sigma_y$-remainder} for short) of $f$ if $f-r$ is
$\sigma_y$-summable and $r$ is $y$-proper with denominator being
$\sigma_y$-free. For brevity, we just say that $r$ is a
$\sigma_y$-remainder if $f$ is clear from the context. Clearly, any
integer shift of a $\sigma_y$-remainder with respect to $x$ is again a
$\sigma_y$-remainder. By \eqref{EQ:ggszred}, we see that the GGSZ
reduction reduces a rational function to a $\sigma_y$-remainder modulo
$\sigma_y$-summable rational functions.

A rational function in $\set K(x,y)$ usually has more than one
$\sigma_y$-remainder and any two of them differ by a
$\sigma_y$-summable rational function. The following proposition
implies that zero is the only $\sigma_y$-remainder in the case of a
$\sigma_y$-summable rational function.
\begin{proposition}[{\citealt[Proposition~7]{Abra1975}}]
  \label{PROP:remainder}
  A rational function in $\set K(x,y)$ is $\sigma_y$-summable if and
  only if any of its $\sigma_y$-remainders is zero.
\end{proposition}

We summarize below the main idea of the reduction-based algorithm
in~\citep{CHKL2015}.

Let $f$ be a rational function in $\set K(x,y)$.  Applying the GGSZ
reduction to $f$ yields \eqref{EQ:ggszred}.  If the denominator of $r$
in \eqref{EQ:ggszred} is not integer-linear, then by
\cite[Theorem~1]{AbLe2002}, $f$ does not have any telescoper.
Otherwise, the existence of telescopers for $f$ is guaranteed.

Assume now that we aim to find a telescoper for $f$ of order no more
than $\rho\in\set N$. In this respect, we make an ansatz
\[
L= c_\rho\shift_x^\rho + \dots + c_1\shift_x + c_0
\quad\text{with $c_0,\dots,c_\rho\in \set K[x]$ to be determined}.
\]
For $\ell=0,\dots,\rho$, compute a rational function $h_\ell\in\set
K(x,y)$ and a $\sigma_y$-remainder $r_\ell$ such that
\begin{equation}\label{EQ:lthggszred}
  \sigma_x^\ell(f) = (\shift_y-1)(h_\ell) +r_\ell
  \quad\text{and}\quad
  \text{$\sum_{i=0}^\ell c_i r_i$ is a $\sigma_y$-remainder}.
\end{equation}
A direct calculation then shows that
\[
L(f)=(\shift_y-1)\left(\sum_{\ell=0}^\rho c_\ell h_\ell\right)
+\sum_{\ell=0}^\rho c_\ell r_\ell.
\]
Therefore, $\sum_{\ell=0}^\rho c_\ell r_\ell$ is a
$\sigma_y$-remainder of $L(f)$. By Proposition~\ref{PROP:remainder},
$L$ is a telescoper for $f$ if and only if $\sum_{\ell=0}^\rho c_\ell
r_\ell=0$.  This reduces the problem of finding telescopers to the
simple task of solving a linear system over $\set K[x]$. In other
words, we obtain a linear homogeneous system in unknowns
$c_0,\dots,c_\rho$ by equating $\sum_{\ell=0}^\rho c_\ell r_\ell$ to
zero, whose any nontrivial solution over $\set K[x]$ gives rise to a
desired telescoper for $f$. Failing to find such a solution implies
that no required telescopers exist.

Again, for computing a minimal telescoper for $f$, the reduction-based
algorithm applies the above process incrementally with
$\rho=0,1,\dots$, with the termination assured by the existence of
telescopers.

The proof of \citep[Theorem~5.6]{CHKL2015} contains an algorithm for
computing such a $\sigma_y$-remainder $r_\ell$ that satisfies
\eqref{EQ:lthggszred}. The key tool is the so-called shift-coprime
decompositions of $\sigma_y$-free polynomials. Let $b,b_0\in\set
K[x,y]$ be two nonzero $\sigma_y$-free polynomials.  The {\em
  $\sigma_y$-coprime decomposition} of $b$ with respect to $b_0$ is
defined as
\begin{equation}\label{EQ:shiftcoprime}
  b = p_0\sigma_y^{\ell_1}(p_1)\cdots\sigma_y^{\ell_m}(p_m),
\end{equation}
where $p_0\in\set K[x,y]$ with $\deg_y(\gcd(b_0,\sigma_y^i(p_0))) = 0$
for any nonzero integer $i$, $p_1,\dots,p_m\in \set K[x,y]$ are monic
and $y$-primitive factors of $b_0$ of positive degrees in $y$, and
$\ell_1,\dots,\ell_m$ are distinct nonzero integers.  Note that the
factors
$p_0,\sigma_y^{\ell_1}(p_1),\dots,\sigma_y^{\ell_m}(p_m),p_1,\dots,p_m$
are pairwise coprime, since $b$ and $b_0$ are both $\sigma_y$-free.
Such a decomposition \eqref{EQ:shiftcoprime} is clearly unique up to
the order of factors.  It is evident from \eqref{EQ:shiftcoprime} and
the $\sigma_y$-freeness of $b$ that
$\disset_y(b,b_0)=\{0,\ell_1,\dots,\ell_m\}$ and $p_i =
\gcd(\sigma_y^{-\ell_i}(b),b_0)$ for all $i = 1,\dots, m$.  Thus the
decomposition \eqref{EQ:shiftcoprime} can be obtained using GCD
computation, provided that the dispersion set $\disset_y(b,b_0)$ is
known.

Let $r,r_0\in\set K(x,y)$ be two nonzero $\sigma_y$-remainders of
respective denominators $b,b_0\in\set K[x,y]$. By partial fraction
decomposition, based on the $\sigma_y$-coprime decomposition
\eqref{EQ:shiftcoprime} of $b$ with respect to $b_0$, there exist
unique $f_0,f_1,\dots,f_m\in\set K(x)[y]$ with
$\deg_y(f_i)<\deg_y(p_i)$ such that
\begin{equation}\label{EQ:scdpfd}
  r = \frac{f_0}{p_0} + \frac{f_1}{\sigma_y^{\ell_1}(p_1)}
  + \dots + \frac{f_m}{\sigma_y^{\ell_m}(p_m)}.
\end{equation}
We will refer to \eqref{EQ:scdpfd} as the {\em SCD-based} partial
fraction decomposition of $r$ with respect to $r_0$.

The following result can be read from the proof of
\cite[Theorem~5.6]{CHKL2015}.
\begin{proposition}\label{PROP:adjustrem}
  Let $r,r_0\in\set K(x,y)$ be two nonzero $\sigma_y$-remainders.
  Assume that the SCD-based partial fraction decomposition of $r$ with
  respect to $r_0$ is given by \eqref{EQ:scdpfd}. Let
  \begin{equation}\label{EQ:adjustrem}
    \tilde r =  \frac{f_0}{p_0} + \frac{\sigma_y^{-{\ell_1}}(f_1)}{p_1}
    +\cdots+\frac{\sigma_y^{-{\ell_m}}(f_m)}{p_m}.
  \end{equation}
  Then $\tilde r$ is a $\sigma_y$-remainder of $r$ and $c_0r_0+c_1\tilde r$
  is a $\sigma_y$-remainder for any $c_0,c_1\in\set K[x]$.
\end{proposition}
In view of the above proposition, we call $\tilde r$ the 
{\em adjusted $\sigma_y$-remainder} of $r$ by  $r_0$.
It then follows from Proposition~\ref{PROP:adjustrem} that
\eqref{EQ:lthggszred} naturally holds by letting $r_\ell$ be the
adjusted $\sigma_y$-remainder of $\sigma_x(r_{\ell-1})$ with respect
to $\sum_{i=1}^{\ell-1}c_ir_i$.  With all these adjusted
$\sigma_y$-remainders at hand, the reduction-based algorithm works
smoothly in an iterative manner as described before.

\begin{remark}\label{REM:rct}
  As already pointed out in \citep[\S 5.2]{CHHLW2019}, it is actually
  sufficient to let each $r_\ell$ be the adjusted $\sigma_y$-remainder
  of $\sigma_x(r_{\ell-1})$ with respect to $r_0$ (rather than
  $\sum_{i=0}^{\ell-1}c_ir_i$) so as to insure the property
  \eqref{EQ:lthggszred}. This may reduce the total cost for computing
  adjusted $\sigma_y$-remainders.
\end{remark}
Let us return to the two examples from Section~\ref{SEC:oct}. We will
use the above reduction-based algorithm in order to illustrate the
difference between the two approaches.
\begin{example}\label{EX:rctsmall}
  Let $f$ be the rational function given in Example~\ref{EX:octsmall}.
  We know from Example~\ref{EX:octsmall} that $f$ has a minimal
  telescoper of order two.  With $\rho=2$, the reduction-based
  algorithm finds the additive decompositions
  \[
  \sigma_x^\ell(f) = (\emph\shift_y-1)(h_\ell) +
  \frac{a_\ell}{b_\ell}\quad\text{for}\ \ell = 0, 1, 2,
  \]
  where $h_\ell\in \set Q(x,y)$, $a_\ell\in\set Z[x,y]$, $b_\ell=
  ((-5x+2y)^2+1)((-5x+2y+1)^2+1)$ and all $a_\ell/b_\ell$, as well as
  their $\set Q[x]$-linear combinations, are $\sigma_y$-remainders.
  Note that the $h_\ell$ and $a_\ell$ are not displayed here for space
  reasons. In order to find a $\set Q[x]$-linear dependency among the
  $a_\ell/b_\ell$, we set up a linear system attached by the
  coefficient matrix
  \begin{equation*}
    \begingroup 
    \setlength\arraycolsep{2pt}
    \begin{pmatrix}
      8x^2+4x & 8x^2+20x+12 & 8x^2+36x+40\\
      -40x^3-12x^2+4 & -40x^3-100x^2-56x & -40x^3-172x^2-168x+36\\
      50x^4+5x^3+4x^2-9x+1 & 50x^4+125x^3+67x^2+6x+3
      & 50x^4+205x^3+174x^2-73x+19
    \end{pmatrix}.
    \endgroup
  \end{equation*}
  
  \noindent This linear system admits the same solutions as
  \eqref{EQ:r1system}, in other words, it leads to the same minimal
  telescoper as Example~\ref{EX:octsmall}. The corresponding
  certificate is left as an unnormalized dense sum.
\end{example}
\begin{example}\label{EX:rct}
  Consider the same rational function $f$ as Example~\ref{EX:ggszred}.
  From the same example, we see that $f$ satisfies \eqref{EQ:ggszred}
  with $h,r$ given by \eqref{EQ:red0}. Moreover, there exist
  telescopers for $f$ since the denominator of $r$ is integer-linear.
  Let $h_0=h$ and $r_0=r$. Then for $\ell = 1,\dots, 22$, the
  reduction-based algorithm iteratively finds rational functions
  $h_\ell\in\set Q(x,y)$ and adjusted $\sigma_y$-remainders $r_\ell$
  such that \eqref{EQ:lthggszred} holds. Finding a $\set Q[x]$-linear
  dependency among the $r_\ell$ yields a linear system with the
  coefficient matrix of 33 rows and 23 columns and having entries of
  degree in $x$ at most 34, which yields the same minimal telescoper
  given by \eqref{EQ:octL} as Example~\ref{EX:modified}, yet leaving
  the corresponding certificate as a large, unnormalized dense
  sum. This compares to Example~\ref{EX:modified} where the induced
  coefficient matrix has 22 rows and 23 columns with entries of degree
  in $x$ at most 2.
\end{example}

\subsection{Output size estimates}
\begin{lemma}\label{LEM:ggszreddeg}
  Let $f\in \set K(x,y)$ be a rational function with $\deg_x(f) = d_x$
  and $\deg_y(f) = d_y$.  Let $r\in \set K(x,y)$ be the
  $\sigma_y$-remainder obtained by applying the GGSZ reduction to $f$.
  Write $r = a/(ub)$, where $u \in \set K[x]$ and $a,b\in \set K[x,y]$
  with $\deg_y(a)<\deg_y(b)$, $\gcd(a,ub) = 1$ and $b$ being
  $y$-primitive and $\sigma_y$-free. Then
  \begin{align*}
    &\deg_x(u)\in \bigO(d_xd_y),\quad
    (\deg_x(b),\deg_y(b)) \in \bigO(d_x)\times \bigO(d_y)\\[1ex]
    \text{and}\quad
    &(\deg_x(a),\deg_y(a)) \in
    \bigO(d_xd_y)\times\bigO(d_y).
  \end{align*}
\end{lemma}
\begin{proof}
  Assume that the denominator $g$ of $f$ admits the shift-homogeneous
  decomposition of the form \eqref{EQ:shiftlessdecomp}.  With respect
  to this, we obtain the unique partial fraction decomposition
  \[
  f = p + \frac1{\tilde u}
  \sum_{i=1}^m\sum_{j=1}^{n_i}\sum_{k=1}^{e_{ij}}
  \frac{f_{ijk}}{\sigma_y^{\nu_{ij}}(g_i)^k},
  \]
  where $p\in \set K(x)[y]$, $\tilde u\in \set K[x]$ and $f_{ijk}\in
  \set K[x,y]$ with $\deg_y(f_{ijk})<\deg_y(g_i)$. Applying
  Fact~\ref{FAC:degreesize} to the above decomposition yields
  $\deg_x(\tilde u) \in \bigO(d_xd_y)$ and $\deg_x(f_{ijk})\in
  \bigO(d_xd_y)$.  Let $d_i = \max_{1\leq j\leq n_i}\{e_{ij}\}$ and
  specify that $f_{ijk} = 0$ in case $k>e_{ij}$. By
  \citep[Theorem~12]{GGSZ2003},
  \[
  r= \frac{a}{ub} = \frac1{\tilde u}\sum_{i=1}^m\sum_{k=1}^{d_i}
  \frac{\sum_{j=1}^{n_i}\sigma_y^{-\nu_{ij}}(f_{ijk})}{g_i^k}.
  \]
  Since $b$ is $y$-primitive, $u$ divides $\tilde u$ in $\set K[x]$
  and thus $\deg_x(u)\in\bigO(d_xd_y)$. Notice that $d_i\leq
  \sum_{j=1}^{n_i}e_{ij}$ for all $i = 1,\dots, m$, so $\deg_y(b) \leq
  \sum_{i=1}^m d_i\deg_y(g_i) \leq d_y$ and similarly, $\deg_x(b)\leq
  d_x$.  Moreover, $\deg_x(a) \leq \max_{ijk}\{\deg_x(f_{ijk})\}+d_x$,
  implying $\deg_x(a)\in\bigO(d_xd_y)$.  The lemma follows.
\end{proof}
\begin{lemma}\label{LEM:lthremainderdeg}
  Let $r=a/(ub)\in\set K(x,y)$ be a $\sigma_y$-remainder, where $u \in
  \set K[x]$ and $a,b\in \set K[x,y]$ with $\deg_y(a)<\deg_y(b)$,
  $\gcd(a,ub) = 1$ and $b$ being $y$-primitive and $\sigma_y$-free.
  Let $\ell\in\set N$ and assume that $r_\ell \in\set K(x,y)$ is a
  $\sigma_y$-remainder of $\sigma_x^\ell(r)$.  Write
  $r_\ell=a_\ell/(u_\ell b_\ell)$, where $u_\ell\in \set K[x]$ and
  $a_\ell,b_\ell\in \set K[x,y]$ with $\deg_y(a_\ell)<\deg_y(b_\ell)$,
  $\gcd(a_\ell,u_\ell b_\ell)=1$ and $b_\ell$ being $y$-primitive and
  $\sigma_y$-free.  Then
  \begin{align*}
    &\deg_x(u_\ell) \leq \deg_x(u)+\deg_x(b)\deg_y(b),
    \quad (\deg_x(b_\ell),\deg_y(b_\ell))=(\deg_x(b),\deg_y(b)),\\[1ex]
    \text{and}\quad
    &(\deg_x(a_\ell),\deg_y(a_\ell)) \leq 
    (\deg_x(a)+\deg_x(b)\deg_y(b),\deg_y(b)-1).
  \end{align*}
\end{lemma}
\begin{proof}
  Since $b_\ell$ is $y$-primitive, it admits the full factorization of
  the form $b_\ell = c_\ell p_1^{e_1}\dots p_m^{e_m}$, where
  $c_\ell\in\set K$ and $p_1,\dots,p_m\in\set K[x,y]\setminus\set
  K[x]$ are distinct, monic and irreducible factors of $b_\ell$ of
  multiplicities $e_1,\dots,e_m$, respectively.  Then by
  \citep[Proposition~5.2]{Huan2016}, $\sigma_x^\ell(b)$ must have the
  form
  \begin{equation}\label{EQ:bfac}
    \sigma_x^\ell(b) = c\,\sigma_y^{k_1}(p_1)^{e_1}
    \cdots \sigma_y^{k_m}(p_m)^{e_m}
    \quad\text{for some}\
    c\in\set K \ \text{and}\ k_1,\dots,k_m\in\set Z.
  \end{equation}
  Consequently, $(\deg_x(b_\ell),\deg_y(b_\ell))=(\deg_x(b),\deg_y(b))$.
  
  On the other hand, notice that $b_\ell$ is $\sigma_y$-free, so
  $\sigma_y^{k_1}(p_1),\dots,\sigma_y^{k_m}(p_m)$ are pairwise coprime.
  Based on the factorization \eqref{EQ:bfac} of $\sigma_x^\ell(b)$, we
  then find unique polynomials $\tilde u\in\set K[x]$ and
  $f_1,\dots,f_m\in\set K[x,y]$ with $\deg_y(f_i)<e_i\deg_y(p_i)$ such
  that
  \begin{equation}\label{EQ:abdecomp}
    \sigma_x^\ell\left(\frac{a}{b}\right) 
    = \frac1{\tilde u}\left(\frac{f_1}{\sigma_y^{k_1}(p_1)^{e_1}}
    + \dots + \frac{f_m}{\sigma_y^{k_m}(p_m)^{e_m}}\right).
  \end{equation}
  Since $r_\ell$ is a $\sigma_y$-remainder of $\sigma_x^\ell(r)$, then
  $r_\ell-\sigma_x^\ell(r)$ is $\sigma_y$-summable. Notice that each
  $f_i/\sigma_y^{k_i}(p_i)^{e_i}$ differs from
  $\sigma_y^{-k_i}(f_i)/p_i^{e_i}$ by a $\sigma_y$-summable rational
  function. We conclude from \eqref{EQ:abdecomp} that
  \begin{equation}\label{EQ:difference}
    \frac{a_\ell}{u_\ell b_\ell}- \frac1{\sigma_x^\ell(u)\tilde u}
    \left(\frac{\sigma_y^{-k_1}(f_1)}{p_1^{e_1}} + \dots
    + \frac{\sigma_y^{-k_m}(f_m)}{p_m^{e_m}}\right)
  \end{equation}
  is $\sigma_y$-summable. Observe that the denominator of the above
  rational function divides $b_\ell$ over $\set K(x)$, so it is
  $\sigma_y$-free.  Since the rational function \eqref{EQ:difference}
  is evidently $y$-proper, it is a $\sigma_y$-remainder by definition.
  It thus follows from Proposition~\ref{PROP:remainder} that
  \eqref{EQ:difference} is equal to zero, that is,
  $a_\ell/(u_\ell b_\ell) = 1/(\sigma_x^\ell(u)\tilde u)
  \sum_{i=1}^m\sigma_y^{-k_i}(f_i)/p_i^{e_i}$.
  Since $b_\ell=c_\ell p_1^{e_1}\dots p_m^{e_m}$ is $y$-primitive,
  $u_\ell$ divides $\sigma_x^\ell(u)\tilde u$ in $\set K[x]$ and then
  $\deg_x(a_\ell)\leq\max_{1\leq i\leq m}\{\deg_x(f_i)+\deg_x(b_\ell)-e_i\deg_x(p_i)\}$.
  The degree estimates for $u_\ell$ and $a_\ell$ thus follow by one
  application of Fact~\ref{FAC:degreesize} (ii) to
  \eqref{EQ:abdecomp}.
\end{proof}

The reduction-based approach also provides us an order-degree curve of
telescopers for bivariate rational functions.
\begin{lemma}\label{LEM:rctorderdegree}
  Let $r=a/(ub)\in\set K(x,y)$ be a $\sigma_y$-remainder, where $u \in
  \set K[x]$ and $a,b\in \set K[x,y]$ with $\deg_y(a)<\deg_y(b)$,
  $\gcd(a,ub) = 1$ and $b$ being $y$-primitive, $\sigma_y$-free and
  integer-linear.  Assume that $b$ admits the refined integer-linear
  decomposition of the form given by the right-hand side of
  \eqref{EQ:refinedildecomp}.  Define
  $\rho_0=\sum_{i=1}^m\mu_i\deg_z(p_i)\max\{e_{i1},\dots,e_{in_i}\}$. Then
  for any nonnegative integer pair $(\rho,\tau)$ with $\rho \geq
  \rho_0$ and
  \begin{align}
    \tau &> \frac{\left((\rho+1)\deg_x(b)\deg_y(b)
      +\rho\deg_x(u)+\deg_x(a)+\rho_0\right)\rho_0
      -\frac12\rho_0(\rho_0-1)-(\rho+1)}{\rho+1-\rho_0},
    \label{EQ:rctdeg}
  \end{align}
  there exists a telescoper for $r$ of order at most $\rho$ and degree
  at most $\tau$.
\end{lemma}
\begin{proof}
  Let $\rho,\tau\in \set N$ with $\rho\geq \rho_0$ and $\tau$
  satisfying \eqref{EQ:rctdeg}. In order to show the lemma, it amounts
  to proving that there exist $c_0,\dots, c_\rho\in \set K[x]$, not
  all zero, with $\deg_x(c_\ell)\leq\tau$ such that
  \begin{equation}\label{EQ:rctsystem}
    c_\rho r_\rho + \dots + c_0 r_0 = 0,
  \end{equation}
  where $r_0 = r$ and $r_\ell$ is the adjusted remainder of
  $\sigma_x(r_{\ell-1})$ by $r_0$ for $\ell = 1,\dots,\rho$, because
  then, by Proposition~\ref{PROP:remainder} and Remark~\ref{REM:rct},
  the operator $\sum_{\ell=0}^\rho c_\ell\shift_x^\ell$ gives a
  desired telescoper for $r$. This then suffices to verify that, for
  the linear homogeneous system over~$\set K$ induced by
  \eqref{EQ:rctsystem}, the number of unknowns, namely
  $(\tau+1)(\rho+1)$ in this case, is greater than the number of
  equations over~$\set K$.  By \citep[Theorem~5.5]{Huan2016} and
  Lemma~\ref{LEM:lthremainderdeg}, the denominator of the left-hand
  side of \eqref{EQ:rctsystem} in $\set K[x,y]$ has total degree in
  $x,y$ at most $(\rho+1)(\deg_x(u)+\deg_x(b)\deg_y(b))+\rho_0$.  By
  separately applying Lemma~\ref{LEM:lthremainderdeg} to
  $r_0,\dots,r_\rho$, one then calculates that there are at most
  \[
  \left(\tau+(\rho+1)\deg_x(b)\deg_y(b)
  +\rho\deg_x(u)+\deg_x(a)+\rho_0\right)\rho_0
  -\frac12\rho_0(\rho_0-1)
  \]
  equations over $\set K$. Since $\rho\geq \rho_0$ and
  \eqref{EQ:rctdeg} holds, a direct comparison between the number of
  unknowns and the above number completes the proof.
\end{proof}
We remark that for \lq\lq generic\rq\rq\ rational functions, $\rho_0$
defined in the above lemma coincides with the one given by
\eqref{EQ:upperbound}, although there are cases in which the latter is
smaller.  Let $f\in\set K(x,y)$ with $\deg_x(f) = d_x$ and $\deg_y(f)
= d_y$ be a rational function admitting $r$ as a $\sigma_y$-remainder.
Lemma~\ref{LEM:rctorderdegree} then asserts that there exists a
minimal telescoper for $f$ of degree in $\bigO(\rho_0^2d_xd_y)$. This
compares to Lemma~\ref{LEM:octorderdegree} which tells us that $f$ can
actually have a minimal telescoper of degree in $\bigO(\rho_0d_xd_y)$.

\subsection{Cost analysis of algorithm}
\begin{lemma}\label{LEM:remainderadj}
  Let $r,r_0\in \set K(x,y)$ be two nonzero $\sigma_y$-remainders.
  Write $r = a/(u b)$ with $u \in \set K[x]$, $a,b\in \set K[x,y]$,
  $\deg_y(a)<\deg_y(b)$, $\gcd(a,u b)=1$ and $b$ being $y$-primitive
  and $\sigma_y$-free. Let $b_0\in \set K[x,y]$ be the $y$-primitive
  denominator of $r_0$. Assume that $\deg_x(b),\deg_x(b_0) \leq d_x$
  and $\deg_y(b),\deg_y(b_0)\leq d_y$. Then the adjusted
  $\sigma_y$-remainder $\tilde r$ of $r$ by $r_0$ can be computed
  using $\bigO(\deg_x(a)^2d_y+d_x^2d_y^3+\deg_x(u)d_xd_y)$ arithmetic
  operations in $\set K$ with classical arithmetic and
  $\softO(\deg_x(a)d_y+d_xd_y^2+\deg_x(u))$ with fast arithmetic, plus
  the cost of computing the dispersion set of $b$ with respect to
  $b_0$.
\end{lemma}
\begin{proof}
  By Proposition~\ref{PROP:adjustrem}, the adjusted
  $\sigma_y$-remainder $\tilde r$ of $r$ by $r_0$ is obtained by
  computing the SCD-based partial fraction decomposition
  \eqref{EQ:scdpfd} of $r$ with respect to $r_0$, along with a
  subsequent normalization based on \eqref{EQ:adjustrem}.  Notice that
  with the dispersion set $\disset_y(b,b_0)$ at hand, computing the
  shift-coprime decomposition of $b$ with respect to $b_0$ merely
  involves GCD computations with arguments of degree in $x$ no more
  than $d_x$ and degree in $y$ no more than $d_y$.  Together with the
  cost of partial fraction decomposition, deriving \eqref{EQ:scdpfd}
  takes $\bigO(\deg_y(a)^2d_y+d_x^2d_y^3)$ arithmetic operations with
  classical arithmetic and $\softO(\deg_y(a)d_y+d_xd_y^2)$ with fast
  arithmetic, plus the cost of computing the dispersion set of $b$
  with respect to $b_0$. Based on
  Facts~\ref{FAC:degreesize}-\ref{FAC:unipolycost}, the final
  normalization of \eqref{EQ:adjustrem} for $\tilde r$ requires
  $\bigO(\deg_y(a)^2d_y+d_x^2d_y^3+\deg_x(u)d_xd_y)$ arithmetic
  operations with classical arithmetic and
  $\softO(\deg_y(a)d_y+d_xd_y^2+\deg_x(u))$ with fast arithmetic. The
  announce cost follows.
\end{proof}
Now we are ready to analyze the cost of the reduction-based creative
telescoping algorithm for bivariate rational functions.
\begin{theorem}\label{THM:rct}
  Let $f\in \set K(x,y)$ be a rational function with $\deg_x(f) = d_x$
  and $\deg_y(f) = d_y$. Assume that $f$ has a telescoper and let
  $\rho$ be the actual order of its minimal telescopers. Further let
  $r\in\set K(x,y)$ be a $\sigma_y$-remainder of $f$, and define
  $\rho_0$ as in Lemma~\ref{LEM:rctorderdegree}. Then the
  reduction-based algorithm in \citep{CHKL2015} finds a minimal
  telescoper for $f$ and an unnormalized certificate using
  $\bigO(d_xd_y^4+\rho
  d_x^2d_y^3+\rho^3\rho_0^3d_x^2d_y^2+\rho\rho_0^5)$ arithmetic
  operations in $\set K$ with classical arithmetic and
  $\softO(d_xd_y^3+\rho
  d_xd_y^2+\rho^3\rho_0^{\omega-1}d_xd_y+\rho^2\rho_0^\omega)$ with
  fast arithmetic, plus the cost of computing the (auto-)dispersion
  sets and finding rational roots.
\end{theorem}
\begin{proof}
  By Lemma~\ref{LEM:ggszred}, the GGSZ reduction takes
  $\bigO(d_xd_y^4+d_x^2d_y^3)$ arithmetic operations with classical
  arithmetic and $\softO(d_xd_y^3)$ with fast arithmetic, plus the
  cost of computing the auto-dispersion set. In addition to the cost
  of finding rational roots in the integer-linearity detection, the
  cost of the remaining algorithm is dominated by computing adjusted
  $\sigma_y$-remainders and solving linear homogeneous systems in
  iteration steps. For the $\ell$-th iteration with $0\leq
  \ell\leq\rho$, by Lemmas~\ref{LEM:ggszreddeg},
  \ref{LEM:lthremainderdeg} and \ref{LEM:remainderadj}, finding the
  $\ell$-th adjusted $\sigma_y$-remainder takes $\bigO(d_x^2d_y^3)$
  with classical arithmetic and $\softO(d_xd_y^2)$ with fast
  arithmetic, plus the cost of computing relevant dispersion
  sets. After this, we need to solve a linear system with the
  coefficient matrix having at most $\rho_0$ rows and $\ell+1$
  columns. Moreover, the entries of the matrix are of degrees in $x$
  in $\bigO(\ell d_xd_y+\rho_0)$. By Fact~\ref{FAC:systemcost},
  finding a solution requires
  $\bigO(\ell^2\rho_0^3d_x^2d_y^2+\rho_0^5)$ with classical arithmetic
  and $\softO(\ell^2\rho_0^{\omega-1}d_xd_y+\ell\rho_0^\omega)$ with
  fast arithmetic. Since there are $\rho$ iterations, this step in
  total takes $\bigO(\rho
  d_x^2d_y^3+\rho^3\rho_0^3d_x^2d_y^2+\rho\rho_0^5)$ with classical
  arithmetic and $\softO(\rho
  d_xd_y^2+\rho^3\rho_0^{\omega-1}d_xd_y+\rho^2\rho_0^\omega)$ with
  fast arithmetic, yielding the announced cost.
\end{proof}

In analogy to Corollary~\ref{COR:oct}, we obtain the following by the
above theorem and Lemma~\ref{LEM:rctorderdegree}.
\begin{corollary}\label{COR:rct}
  With the assumptions of Theorem~\ref{THM:rct}, further let $\mu =
  \max\{\mu_1,\dots,\mu_m\}$. Then, without expanding the certificate,
  the reduction-based algorithm in \citep{CHKL2015} takes
  $\bigO(\mu^6d_x^2d_y^8)$ arithmetic operations in~$\set K$ with
  classical arithmetic and $\softO(\mu^{\omega+2}d_xd_y^{\omega+3})$
  with fast arithmetic, plus the cost of computing (auto-)dispersion
  sets and finding rational roots.
\end{corollary}
\begin{proof}
  It is evident from the definition of $\rho_0$ that
  $\rho_0\in\bigO(\mu d_y)$.  By Lemma~\ref{LEM:rctorderdegree}, $\rho
  \leq \rho_0$ since $\rho$ is the actual order of minimal telescopers
  for $f$. Thus $\rho\in\bigO(\mu d_y)$.  The announced cost then
  directly follows by Theorem~\ref{THM:rct}.
\end{proof}
The above result compares to Corollary~\ref{COR:oct} which announces
that for the same purpose, the algorithm {\bf RationalCT} takes
$\bigO(\mu^4d_x^2d_y^6)$ arithmetic operations in~$\set K$ with
classical arithmetic and $\softO(\mu^{\omega+1}d_xd_y^{\omega+2})$
with fast arithmetic, plus the cost of computing auto-dispersion sets
and finding rational roots.

Note that for a polynomial $b\in \set K[x,y]$, computing its
auto-dispersion set and computing the dispersion set
$\disset_y(\sigma_x(b),b)$ take almost the same cost. Hence the extra
costs for the two algorithms in fact do not differ too much.

\section{Implementation and timings}\label{SEC:timing}
We have implemented our algorithms in the computer algebra system
{\sc Maple~2018}. Our implementation includes the four enhancements
discussed in Section~\ref{SUBSEC:modifications}. The code is available
by email request.  In order to get an idea about the efficiency, we
compared their running time and memory requirements to the performance
of two known algorithms -- the one developed by \cite{Le2003a} and the
reduction-based one reviewed in Section~\ref{SEC:rct}. The
implementation for the former algorithm uses the built-in Maple
procedure {\bf SumTools[Hypergeometric][ZpairDirect]}, while the
implementation for the latter algorithm was done in accordance with
descriptions of the algorithm {\bf ReductionCT} from \citep{CHKL2015}
restricted to the rational case, by embracing the GGSZ reduction and
Remark~\ref{REM:rct}.  All timings are measured in seconds on a Linux
computer with 128GB RAM and fifteen 1.2GHz Dual core processors. The
computations for the experiments did not use any parallelism.

We take examples of the expanded form of
\begin{equation}\label{EQ:testsuite}
  r(x,y) = (\shift_y-1)\left(\frac{f_0(x,y)}{g_0(x,y)}\right)
  +\frac{f(x,y)}{g_1(-\lambda x+\mu y)\cdot g_2(\lambda x+\mu y)},
\end{equation}
where
\begin{itemize}
\item $f_0,f\in \set Z[x,y]$ of total degree $m\geq 0$ and max-norm
  $||f_0||_\infty, ||f||_\infty\leq 20$;
\item $g_0\in \set Z[x,y]$ of total degree $n\geq 0$ and max-norm
  $||g_0||_\infty \leq 20$;
\item $\lambda,\mu$ are positive integers;
\item $g_i\in \set Z[z]$ of the form $g_i = p_i(z)p_i(z+\lambda_i)
  p_i(z+\lambda_i\mu) p_i(z+\lambda_i+\lambda_i\mu)$ for $\lambda_i =
  (-1)^i\lambda$ and $p_i\in \set Z[z]$ of total degree $n>0$ and
  max-norm $||p_i||_\infty \leq 20$.
\end{itemize}
Note that in a generic situation, a rational function $r\in \set
Q(x,y)$ of the form \eqref{EQ:testsuite} admits the following
RILD-based partial fraction decomposition
\[
r = (\shift_y-1)\left(\frac{f_0(x,y)}{g_0(x,y)}\right)
+ M_1\left(\frac1{p_1(-\lambda x + \mu y)}\right)
+ M_2\left(\frac1{p_2(\lambda x + \mu y)}\right),
\]
where $M_i = a_{i0} + a_{i1}\shift_{\lambda_i,\mu}^{\lambda_i} +
a_{i2}\shift_{\lambda_i,\mu}^{\lambda_i\mu} +
a_{i3}\shift_{\lambda_i,\mu}^{\lambda_i+\lambda_i\mu}$ for some
$a_{i0},a_{i1},a_{i2},a_{i3}\in \set Q(x)[y]$. As such, by modulo some
$\sigma_y$-summable rational function, it can be further reduced to
\[
R_1\left(\frac1{p_1(-\lambda x + \mu y)}\right)
+ R_2\left(\frac1{p_2(\lambda x + \mu y)}\right)
\]
with $R_i = b_{i0}+b_{i1}\shift_{\lambda_i,\mu}^{\lambda_i}$ for some
$b_{i0},b_{i1}\in \set Q(x)[y]$.

For a selection of random rational functions of this type for
different choices of $(m,n,\lambda,\mu)$, Table~\ref{TAB:timing}
collects the timings, without expanding the certificate, of the
algorithm of Le ({\sf DCT}), the reduction-based algorithm ({\sf RCT})
and our algorithm ({\sf OCT}) developed in Section~\ref{SEC:oct}.  The
column {\em order} is used to record the actual order of the output
minimal telescoper.

\begin{table}[ht]
  \centering
  \begin{tabular}{l|r|r|r|c}
    $(m,n,\lambda,\mu)$ & {\sf DCT} & {\sf RCT} & {\sf OCT}
    & order \\\hline
    (1, 1, 1, 1) & 0.18 & 0.17 & 0.16 & 2\\
    (1, 1, 4, 1) & 0.18 & 0.20 & 0.16 & 2\\
    (1, 1, 16, 1) & 0.19 & 0.21 & 0.17 & 2\\
    (5, 1, 4, 1) & 0.22 & 0.23 & 0.19 & 3\\
    (10, 1, 4, 1) & 0.26 & 0.27 & 0.21 & 3\\
    (15, 1, 4, 1) & 0.46 & 0.40 & 0.27 & 4\\
    (15, 1, 4, 5) & 10.43 & 14.63 & 0.90 & 10\\
    (15, 1, 4, 7) & 46.39 & 69.64 & 1.92 & 14\\
    (15, 1, 4, 9) & 181.34 & 283.65 & 3.58 & 18\\
    (15, 1, 4, 11) & 456.69 & 851.72 & 7.49 & 22\\
    (15, 1, 4, 13) & 892.44 & 2436.57 & 13.59 & 26\\
    (1, 2, 4, 1) & --~~ & 15.24 & 2.48 & 7\\
    (1, 3, 4, 1) & --~~ & 1220.58 & 49.19 & 11\\
    (1, 4, 4, 1) & --~~ & 30599.21 & 935.41 & 15\\
    (10, 2, 4, 1) & --~~ & 21.00 & 3.96 & 7\\
    (20, 2, 4, 1) & --~~ & 27.27 & 5.92 & 7\\
    (30, 2, 4, 1) & --~~ & 51.82 & 14.55 & 8\\
    (30, 2, 4, 3) & --~~ & 504.78 & 51.93 & 12\\
    (30, 2, 4, 5) & --~~ & 6437.51 & 436.25 & 20\\
    (30, 2, 4, 7) & --~~ & 47763.39 & 1283.01 & 28\\
    \hline
  \end{tabular}
  \caption{Comparison of three algorithms for a collection of rational
    functions of the form \eqref{EQ:testsuite}.}
  \label{TAB:timing}
\end{table}

From the finding we see that our creative telescoping algorithm has
comparable timings for random problems of small size. In particular
none of the three algorithms have significant set up costs.  As $m$
increases our algorithm shows significant improvement over both the
direct and reduction-based methods. The dash in the column {\sf DCT}
indicates that the current built-in procedure for {\sf DCT} in
{\sc Maple 2018} is not applicable for random inputs with this choice
of $(m,n,\lambda,\mu)$.  The issue in these cases is that the
denominator of the input rational function has irreducible factors of
degrees greater than one, and then the algorithm of \cite{Le2003a}
requires recurrence operators with coefficients being polynomials over
algebraic numbers, something not yet included in the current
implementation of {\sf DCT} in Maple.

\section{Conclusion and future work}\label{SEC:conclusion}
A new algorithm of creative telescoping for bivariate rational
functions has been developed in this paper. Our algorithm is based on 
basic arithmetic in the ring of recurrence operators and expresses the
certificate part by a compact representation, which, if desired, can
be expanded in time polynomial in the size of the final result. In
terms of complexity, our algorithm outperforms the reduction-based
approach in the case of bivariate rational functions by at least one
order of magnitude ignoring the certificate part. In practice, our
algorithm is also more efficient according to the experiments.

With the rational case being settled, it is natural to wonder about an
analogous algorithm for hypergeometric terms. Recall that a bivariate
function $f(x,y)$ is called a {\em hypergeometric term} if both
$f(x+1,y)/f(x,y)$ and $f(x,y+1)/f(x,y)$ are rational functions in
$x,y$. The hypergeometric term is a basic and ubiquitous class of
special functions appearing in combinatorics~\citep{PWZ1996}. It is
more interesting and also more challenging than the rational case.

In the hypergeometric case, there exists no direct analog of the
partial fraction decomposition of rational functions. Thus the method
described in this paper will not work directly for this setting. One
possible way to proceed is to first compute a multiplicative
decomposition of the given hypergeometric term and then reduce the
problem to a rational one (cf.\ \citep{AbPe2001a,CHKL2015}).  This
way, however, may introduce arithmetic operations on recurrence
operators over $\set K(x,y)$ instead of $\set K(x)[y]$, and thus makes
it more difficult to derive a hypergeometric telescoping criterion,
namely an analog of Theorem~\ref{THM:criterion}. In the future, we
hope to explore this topic further and aim at generalizing our results
to the class of hypergeometric terms and beyond.

\section*{Acknowledgments}
We would like to express our gratitude to Ziming Li for his helpful
discussions and valuable comments, which improved this work
considerably. We also would like to thank the anonymous referees for
many useful and constructive suggestions. Most of the work presented
in this paper was carried out while Hui Huang was a Post Doctoral
Fellow at the University of Waterloo. 
This research was partly supported by the Natural Sciences and Engineering 
Research Council (NSERC) Canada (No.\ NSERC RGPIN-2018-04950, 
No.\ NSERC RGPIN-2020-04276 and No.\ NSERC RGPIN 238778-06). 
Hui Huang was also supported by the Fundamental Research Funds 
for the Central Universities (No.\ DUT20RC(3)073).

\bibliographystyle{elsarticle-harv}
\newcommand{\Gathen}{\relax}\newcommand{\Hoeven}{\relax}\def\cprime{$'$}
\def\cprime{$'$} \def\cprime{$'$} \def\cprime{$'$} \def\cprime{$'$}
\def\cprime{$'$} \def\cprime{$'$} \def\cprime{$'$}
\def\polhk#1{\setbox0=\hbox{#1}{\ooalign{\hidewidth
			\lower1.5ex\hbox{`}\hidewidth\crcr\unhbox0}}} \def\cprime{$'$}

\end{document}